\documentclass[11pt,reqno]{amsart}
\usepackage{amsmath,amssymb}
\usepackage{amsthm} 
\usepackage{amscd}

\usepackage[pdftex]{graphicx}%%% arXiv—p
\usepackage[all]{xy}

\makeatletter
\@addtoreset{equation}{section}
\makeatother
\usepackage{enumerate}

\def\be{\mathbf e}

\def\CC{{\mathbb C}}

\def\LL{{\mathbb L}}
\def\RR{{\mathbb R}}

\def\PP{{\mathbb P}}

\def\be{{\mathbf e}}

\def\ee{{\mathrm e}}

\def\ii{{\sqrt{-1}}}

\def\cB{{\mathcal B}}
\def\cC{{\mathcal C}}

\def\cG{{\mathcal G}}

\def\cK{{\mathcal K}}
\def\cL{{\mathcal L}}
\def\cM{{\mathcal M}}

\def\cV{{\mathcal V}}

\def\fb{{\mathfrak b}}

\def\fs{{\mathfrak s}}

\def\ri{{\mathrm i}}
\def\rr{{\mathrm r}}

\def\tA{{\widetilde A}}

\def\tK{{\widetilde K}}

\def\tX{{\widetilde X}}

\def\tS{{\widetilde S}}

\def\tU{{\widetilde U}}

\def\tX{{\widetilde X}}

\def\tv{{\widetilde v}}

\def\tgamma{{\widetilde \gamma}}

\def\tSS{{{\widetilde{\mathbb S}}}}

\def\hv{{\widehat v}}

\def\hB{{\widehat B}}

\def\hJ{{\widehat J}}

\def\hX{{\widehat X}}

\def\hkappa{{\widehat \kappa}}

\def\hnu{{\widehat \nu}}

\def\hvarpi{{\widehat \varpi}}

\def\hGamma{{\widehat \Gamma}}

\def\cn{{\mathrm {cn}}}
\def\sn{{\mathrm {sn}}}
\def\dn{{\mathrm {dn}}}

\def\al{{\mathrm {al}}}

\def\qed{\hbox{\vrule height6pt width3pt depth0pt}}
\def\nuI#1{{\nu_{#1}}}

\def\trp{{\, {}^t\negthinspace}}

\def\book#1{\rm{#1}, }
\def\paper#1{\textit{#1}, }
\def\jour#1{\rm{#1}, }
\def\yr#1{({\rm{#1}) }}
\def\vol#1{\textbf{#1}}
\def\pages#1{\rm{#1}}

\def\publaddr#1{\rm{#1}, }
\def\publ#1{\rm{#1}, }
\def\by#1{{\rm{#1}, }}

\newtheorem{theorem}{Theorem}[section]

\newtheorem{proposition}[theorem]{Proposition}

\newtheorem{remark}[theorem]{Remark}
\newtheorem{lemma}[theorem]{Lemma}
\newtheorem{assumption}[theorem]{Assumption}

\newtheorem{condition}[theorem]{Condition}

\def\book#1{\rm{#1}, }
\def\paper#1{\textit{#1}, }
\def\jour#1{\rm{#1}, }
\def\yr#1{({\rm{#1}) }}
\def\vol#1{\textbf{#1}}
\def\pages#1{\rm{#1}}

\def\publaddr#1{\rm{#1}, }
\def\publ#1{\rm{#1}, }
\def\by#1{{\rm{#1}, }}

\begin{document}

%\begin{multicols}{2}

\title[Closed plane curves of hyperelliptic solutions of GKdV equation]{Closed real plane curves of hyperelliptic solutions of focusing gauged modified KdV equation of genus three}

%Electrical Engineering and Computer Science,
\author{Shigeki Matsutani}
% \affiliation{
%1) Graduate School of Natural Science \& Technology, 
%Kanazawa University Kakuma Kanazawa, 920-1192, Japan}
%
%\email{s-matsutani@se.kanazawa-u.ac.jp}

\date{\today}

\begin{abstract}
The real and imaginary parts of the focusing modified Korteweg-de Vries (MKdV) equation defined over the complex field $\CC$ give rise to the focusing gauged MKdV (FGMKdV) equations.
As a generalization of Euler's elastica whose curvature obeys the focusing static  MKdV (FSMKdV) equation, we study real plane curves whose curvature obeys the FGMKdV equation since the FSMKdV equation is a special case of the FGMKdV equation.
In this paper, we focus on the hyperelliptic curves of genus three.
By tuning some moduli parameters and initial conditions, we show closed real plane curves associated with the FGMKdV equation beyond Euler's figure-eight of elastica.
\end{abstract}

\maketitle
{\bf{Keywords:}}
{elastica, modified KdV equation, gauged modified KdV equation, real hyperelliptic solutions, hyperelliptic curves, focusing gauged modified KdV equation}

\section{Introduction}\label{sec:1}

Euler's elastica is a typical real plane curve $Z: [0,1] \to \CC$ whose curvature $k$ obeys the focusing static modified Korteweg-de Vries (FSMKdV) equation,
\begin{equation}
a\partial_{s}k
           +\frac{3}{2}k^2 \partial_s k
+\partial_{s}^3 k=0, \quad
a\partial_{s}\phi
           +\frac{1}{2}(\partial_s \phi)^3
+\partial_{s}^3 \phi=0,
\label{4eq:SMKdV_k}
\end{equation} 
where $\partial_s := \partial/\partial s$, $s$ is the arclength and $\phi := \log \partial_s Z/\ii$ is the tangential angle of the curve i.e., the curvature $k = \partial_s \phi$.
We note that $|\partial_s Z|=1$.
The solutions of (\ref{4eq:SMKdV_k}) are well-written in terms of the theory of the elliptic function:
$Z$ is expressed by the Weierstrass zeta function as in \cite{Mat10}.

Euler found the shapes of $Z$ as the solution of the minimal problem of the energy $\displaystyle{\int k^2 ds}$, so-called the Euler-Bernoulli energy functional under the isometric (non-stretching) condition.
He completely classified the shapes of elasticae numerically and provided the drawings of elasticae to which we refer as Euler's list.
Considering the case $S^1$, i.e. $Z:S^1 \to \CC$ or the closed curves, Euler concluded that there are only two cases, the circle and the figure-eight, in the list, as in Figure \ref{fg:shape03} for the solution of (\ref{4eq:SMKdV_k}) \cite{Euler44}.

This paper provides a generalization of Euler's figure-eight of elliptic curve of genus one to that of hyperelliptic curve of genus three.
We show it as the focusing gauged MKdV flow as in Figures \ref{fg:shape01} and  \ref{fg:shape02}.

\begin{figure}
\begin{center}

\includegraphics[height=0.5\hsize, angle=90]{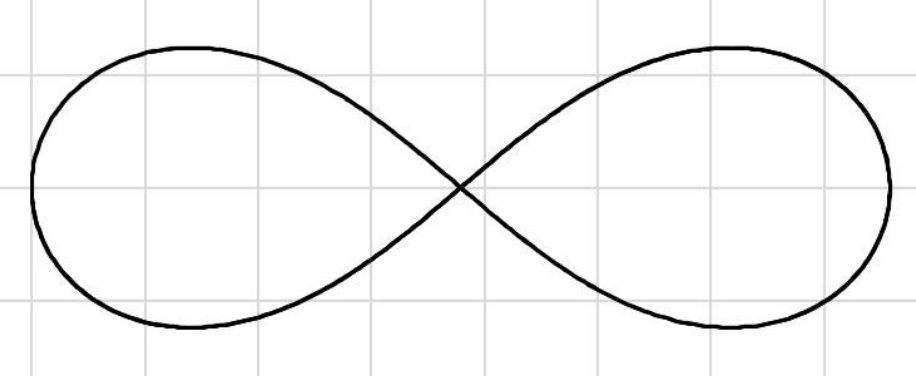}

\end{center}

\caption{
Euler's figure-eight
}\label{fg:shape03}
\end{figure}

There are several generalizations of Euler's elastica \cite{JGaray, Koiso, Pinkall}.
In this paper we are interested in the modified KdV flow as one of its natural generalization.
The modified Korteweg-de Vries (MKdV) equation,
\begin{equation}
\partial_t q \pm 6q^2 \partial_s q +\partial_s^3 q =0,
\label{eq:MKdV1}
\end{equation}
where $\partial_t := \partial/\partial t$ for the real axes $t$ and $s$, is referred to as the focusing MKdV equation for the \lq\lq$+$\rq\rq case in (\ref{eq:MKdV1}) and  is called the defocusing MKdV equation for the \lq\lq$-$\rq\rq case \cite{ZakharovShabat}. 

Following a study on an integrable differential equation of plane curves by Konno, Ichikawa and Wadati \cite{KIW, KIW2}, Ishimori showed that the plane curves whose curvature $k$ obeys the focusing MKdV equation (\ref{eq:MKdV1})
\begin{equation}
\partial_t \phi
           +\frac{1}{2}(\partial_s \phi)^3
+\partial_{s}^3 \phi=0,
\label{4eq:MKdV1phi}
\end{equation}
where $q=k/2=\partial_s \phi/2$ can be regarded as a generalization of Euler's elastica (\ref{4eq:SMKdV_k}) \cite{Ishimori}. Konno and Jeffrey provided two soliton solutions  \cite{KonnJeffrey}.
They called it loop soliton.
When $s=t/a$,  (\ref{4eq:MKdV1phi}) is reduced to (\ref{4eq:SMKdV_k})

Independently, Goldstein and Petrich showed that the isometric deformation of a real curve on a plane is connected with the recursion relations of the focusing MKdV hierarchy, which is also a generalization of Euler's elastica (\ref{4eq:SMKdV_k}) \cite{GoldsteinPetrich1}.
After \cite{GoldsteinPetrich1}, the MKdV flow has been studied by several authors \cite{CKPP, Langer, LangerPerline, MP16, P}.

We found that the Goldstein-Petrich scheme \cite{GoldsteinPetrich1} plays an essential role in the statistical mechanics of the elasticae \cite{Mat97}.
In other words, we showed that the solutions of the MKdV equation can be regarded as equi-energy orbits of the Euler-Bernoulli energy functional under the isometric condition, i.e., the excitation states of elastica.
As in \cite{Ma24a}, finding the hyperelliptic solutions of the focusing MKdV equation is critical to reproducing the shapes of supercoiled DNA observed in laboratories, since the shapes of DNA are given by the excitation states of the Euler-Bernoulli energy functional as thermal effects.
It demonstrates a fascinating and beautiful relationship between modern mathematics and the life sciences.

However, in this paper, we will only deal with the solutions of the focusing gauge MKdV (FGMKdV) equation rather than the focusing MKdV equation itself.
To show the relationship between the two, we explain their background as follows.

\bigskip

For a hyperelliptic curve $X_g$ given by $y^2 + (x-b_0)(x-b_1)\cdots(x-b_{2g})=0$ for $b_i \in \CC$, due to Baker \cite{Baker97, Baker03, BEL97b, M24}, we find the hyperelliptic solution of the KdV equation as $\wp_{gg}(u)=x_1+ \cdots +x_g$ for $((x_1, y_1), \ldots, (x_g, y_g)) \in S^g X_g$ ($g$-th symmetric product of $X_g$) as a function of $u \in \CC^g$ through the Abel-Jacobi map $v: S^g X_g \to J_X$ for the Jacobi variety $J_X$, $u=v((x_1,y_1), \ldots, (x_g, y_g))$.
With the help of the Miura map it is not difficult to find the hyperelliptic solutions of the focusing MKdV equation over $\CC$ \cite{Mat02c}, i.e.,
\begin{equation}
\partial_{u_{g-1}}\psi
-\frac{1}{2}(\lambda_{2g}+3b_0) \partial_{u_g}\psi
+\frac{1}{8}
\left(
\partial_{u_g} \psi\right)^3
+\frac{1}{4}\partial_{u_g}^3 \psi = 0,
\label{1eq:fprMKdV}
\end{equation}
where  $\psi:=\log( (b_0-x_1)\cdots(b_0-x_g))/\ii$, 
and $\lambda_{2g}=-\sum_i b_i$.

Let $u_a \in \CC^g$ and $\psi$ be decomposed to its real and imaginary parts, $u_a = u_{a\,\rr} + \ii u_{a\, \ri}$, $\partial_{u_a}=\frac{1}{2}(\partial_{u_a\,\rr}-\ii \partial_{u_a\,\ri})$, $(a=1, \ldots, g)$, and $\psi=\psi_\rr + \ii \psi_\ri$ of the real valued functions  $\psi_\rr$ and $\psi_\ri$.
The Cauchy-Riemann relations mean $\partial_{u_a\,\rr} \psi_{\rr}=\partial_{u_a\,\ri} \psi_{\ri}$ and $\partial_{u_a\,\rr} \psi_{\ri}=-\partial_{u_a\,\ri} \psi_{\rr}$, and thus $\partial_{u_a}\psi = \partial_{u_a\,\rr} \psi_{\rr}-\ii \partial_{u_a\,\ri} \psi_{\rr}$ or $\partial_{u_a}\psi = \partial_{u_a\,\rr} (\psi_{\rr}+\ii \psi_{\ri})$ $(a=1,\ldots,g)$. 
Since (\ref{1eq:fprMKdV}) contains the cubic term $(\partial_{u_g} \psi)^3=(\partial_{u_g\,\rr} \psi)^3$, it generates the term $-3(\partial_{u_g\, \rr} \psi_\ri)^2  \partial_{u_g\, \rr} \psi_\rr$, which behaves like a gauge potential.
Thus we encounter a differential relation from the focusing MKdV equation over $\CC$ (\ref{1eq:fprMKdV}) as
\begin{equation}
(\partial_{u_{g-1}\, \rr}-
A(u)\partial_{u_g\, \rr})\psi_\rr
           +\frac{1}{8}
\left(\partial_{u_g\, \rr} \psi_\rr\right)^3
+\frac{1}{4}\partial_{u_g\, \rr}^3 \psi_\rr=0,
\label{1eq:gaugedMKdV2}
\end{equation}
where $A(u)=(\lambda_{2g}+3b_a+\frac{3}{4}(\partial_{u_{g}\, \rr}\psi_\ri)^2)/2$.
We refer to it as the focusing gauged MKdV (FGMKdV) equation as mentioned above.
We note the FSMKdV equation (\ref{4eq:SMKdV_k}) is the special case of (\ref{1eq:gaugedMKdV2}) since $\partial_{u_{g}\, \rr}\psi_\ri$ vanishes for $g=1$ case.

\cite{MP22} shows that obtaining the real solution of focusing MKdV equation (\ref{4eq:MKdV1phi}) is to find the situation where the following conditions are satisfied for the solutions of (\ref{1eq:gaugedMKdV2}):
\begin{enumerate}

\item[CI] $\prod_{i=1}^g |x_i - b_a|=$ a constant $(> 0)$,

\item[CII] $d u_{g\,\ri}=d u_{g-1\, \ri}=0$ or $d u_{g\,\rr}=d u_{g-1\, \rr}=0$, and

\item[CIII] $A(u)$ is a real constant:
if $A(u)=$ constant, (\ref{1eq:gaugedMKdV2}) is reduced to (\ref{4eq:MKdV1phi}), i.e., $\psi_\rr=\phi$. 
\end{enumerate}

Although it is quite difficult to find the real plane $\{(u_{g,\rr}, u_{g-1,\rr})\}$ in the Jacobian $J_X$ corresponding to the unit circle valued $(b_0-x_1)\cdots(b_0-x_g)\in \mathrm{U}(1)$, in the previous paper \cite{Ma24b}, we found a nice parameterization of the Jacobian $J_X$ to capture the data of the real hyperelliptic solutions of the FGMKdV equation.
Using the results in \cite{Ma24b}, we showed the relationship between the solutions of the focusing MKdV equation and the shapes of the supercoiled DNA in \cite{Ma24a}.

Our strategy is as follows:
First we find the real hyperelliptic solutions of the FGMKdV equation, and then we evaluate the gauge field $A$.
If $A$ is constant or nearly constant, the solutions are interpreted as the hyperelliptic function solutions of the focusing MKdV equation over $\RR$.

In this paper we adopt this strategy and mainly consider the hyperelliptic solutions of the FGMKdV equation based on the previous investigation \cite{Ma24b} rather than solutions of the focusing MKdV equation over $\RR$.
As in \cite{Ma24b}, we focus on the hyperelliptic curves of genus three.
Since in previous papers \cite{Ma24a, Ma24b}, we could handle the data of the real hyperelliptic solutions of the FGMKdV equation, we believe that our novel method in \cite{Ma24b}, which has never been seen in related study, has advantage to obtain the closed real plane curves beyond Euler's figure-eight, and thus we will advance it in this paper.

\bigskip

However, in the previous papers \cite{Ma24a, Ma24b}, we could not explicitly find the hyperelliptic solutions even of the FGMKdV equation because our parameterization of the Jacobian $J_X$ in \cite{Ma24b} needed improvement, although we dealt with the data of the real hyperelliptic solutions.

In this paper we revise the parameterization in the previous paper \cite{Ma24b}  to find a nicer real parameterization of the Jacobian $J_X$ than \cite{Ma24b} to express the real hyperelliptic solutions of the FGMKdV equation as in Theorem \ref{th:4.2}.
Thus, in Section 4 we will improve and modify several quantities in \cite{Ma24b} to obtain such nice properties.
After we find solutions $\psi$ of the FGMKdV equation (\ref{1eq:gaugedMKdV2}) for $s = u_{g,\rr}$, we define $Z=\int^s \ee^{\ii \psi}\,ds$ which can be regarded as a generalization of Euler's elastica because the SMKdV equation (\ref{4eq:SMKdV_k}) is a special case of the FGMKdV equation (\ref{1eq:gaugedMKdV2}) as mentioned above.
By tuning some moduli parameters and initial conditions, we numerically provide  closed plane curves $\{Z\}$ associated with the FGMKdV equation beyond Euler's figure-eight of elastica in Section 5.

\bigskip

The contents are as follows:
Although Sections 2 and 3 are basically the same as parts in the previous paper \cite{Ma24b}, we briefly explain them for the sake of self-containment in this paper.
Section 2 reviews the hyperelliptic solutions of the focusing MKdV equation over $\CC$ of genus three following \cite{Mat02b,MP15,MP22}.
Section 3 is devoted to the angle expression of the hyperelliptic curves of genus three following \cite{Ma24b}.
Section 4 is a key section in this paper.
It explicitly provides local properties of the real solutions of the FGMKdV equation (\ref{1eq:gaugedMKdV2}) as in Theorem \ref{pr:solgMKdV};
There by revising the situations written in the previous paper \cite{Ma24b}, we finally obtain a nice real parameterization of the hyperelliptic solutions of the FGMKdV equation as in Theorem \ref{th:4.2}.
Using Theorem  \ref{th:4.2}, we have the global solution of the FGMKdV equation as in Theorem \ref{pr:solgMKdV}.
Following the proof of Theorem \ref{pr:solgMKdV}, we numerically investigate the global behavior of the hyperelliptic solutions of genus three of the FGMKdV equation in Section 5.
Section 5 provides the numerical results of the closed curves on a plane whose curvature obeys the FGMKdV equation by tuning the moduli parameters of the hyperelliptic curves.
In Section 6, as a discussion and conclusion, we will mention remarks on our results and open problems for future studies.

\section{Hyperelliptic solutions of the focusing MKdV equation$/\CC$ of genus three}
\label{sec:HESGE}

We review the hyperelliptic solutions of the focusing MKdV equation over $\CC$ of genus three as mentioned in \cite{Ma24b}.
Although this section is basically the same as Section 3 in \cite{Ma24b}. we explain the geometric setting of hyperelliptic curves for the sake of self-containment in this paper.

For a solution of (\ref{4eq:MKdV1phi}) based on hyperelliptic function theory, we should consider hyperelliptic curves  of genus three. 
Here, we handle a hyperelliptic curve $X$ of genus three over $\CC$,
\begin{equation}
X=\left\{(x,y) \in \CC^2 \ |
\ y^2 + (x-b_0)(x-b_1)(x-b_2)\cdots(x-b_{6})=0\right\}
\cup \{\infty\},
\label{4eq:hypC}
\end{equation}
where $b_i$'s are mutually distinct complex numbers.
Let $\lambda_{6}=\displaystyle{-\sum_{i=0}^{6} b_i}$ and $S^k X$ be the $k$-th symmetric product of the curve $X$. 
Further, we introduce the Abelian covering $\tX$ of $X$ by abelianization of the path-space of $X$ divided by the homotopy equivalence, $\kappa_X: \tX \to X$, ($\gamma_{P, \infty} \mapsto P$) \cite{M24}.
Here $\gamma_{P, \infty}$ means a path from $\infty$ to $P$.
$S^k \tX$ also means the $k$-th symmetric product of the space $\tX$. 
The Abelian integral $\tv:=\displaystyle{\begin{pmatrix} v_1\\ v_2\\ v_3\end{pmatrix}} : S^3  \tX \to \CC^3$ is defined by
\begin{equation}
\tv_i(\gamma_1,\gamma_2,\gamma_3)=\sum_{j=1}^3
 \tv_i(\gamma_j), \quad
\tv_i(\gamma_{(x,y), \infty}) = \int_{\gamma_{(x,y), \infty}} \nuI{i},\quad
\nuI{i} = \frac{x^{i-1}d x}{2y}.
\label{4eq:firstdiff}
\end{equation}
Then we have the Jacobian $J_X$, $\kappa_J: \CC^3 \to J_X=\CC^3/\Gamma_X$, where $\Gamma_X$ is the lattice generated by the period matrix for the standard homology basis of $X$.
Due to the Abel-Jacobi theorem \cite{FarkasKra}, we also have the bi-rational map $v$ from $S^3 X$ to $J_X$ by letting $v:=\tv$ modulo $\Gamma_X$.
We refer to $v$ as the Abel-Jacobi map.

\bigskip
Further, we introduce a double covering $\hX$ of $X$ as in Figure \ref{fg:Fig00} by considering $\al_a(u)$ function $\sqrt{\prod_{i=1}^3(b_a-x_i)}$ for a branch point $B_a:=(b_a, 0) \in X$ ($\varpi_x: X\to \PP^1$) and $((x_i, y_i))_i \in S^3 X$, which is originally defined by Weierstrass 1854 \cite{Wei54, Baker98, M24}.
Fix $a=0$.
The square root leads the transformation of $w^2 = (x-b_0)$, i.e., the double covering $\hX$ of the curve $X$, $\varpi_\hX: \hX \to X$, although the precise arguments are left to the Appendix in \cite{MP15}.
Since $\tX$ is also a covering of $\hX$, we have a natural commutative diagram,
\begin{equation}
\xymatrix{ 
 \tX \ar[dr]^{\kappa_X}\ar[r]^-{\kappa_\hX}& \hX \ar[d]^-{\varpi_\hX} \\
  & X.
}\label{2eq:al_hyp_kappas}
\end{equation}
The $\al_a(u)$ function is a generalization of the Jacobi $\sn, \cn, \dn$ functions because the Jacobi function consists of $\sqrt{x-e_i}$, ($i=1,2,3$) of genus one for a curve $y^2 = \prod_{i=1}^3(x-e_i)$.

The curve $\hX$ is given by $z^2 + (w^2-e_1)\cdots(w^2 - e_{6})=0$, where $z:=y/w$ (due to normalization), and $e_j := b_j- b_0$, $j=1, \ldots, 6$.
Since the genus of $\hX$ is five, we have five holomorphic one-forms,
$$
\hnu_j:= \frac{w^{j} d w}{z}, \quad (j=1, 2, 3, \ldots, 5).
$$
and the Jacobi variety $J_{\hX}$ of $\hX$ is given by the complex torus $J_{\hX}=\CC^5 /\Gamma_{\hX}$ for the lattice $\Gamma_{\hX}$ given by the period matrix.
As in \cite[Appendix, Proposition 11.9]{MP15}, we have the correspondence $\varpi_X^*\nuI{i}=\hnu_{2i-2}$, $(i=1,2,3)$ and thus the Jacobian $J_{\hX}$  contains a subvariety $\hJ_X\subset J_{\hX}$ which is a double covering of the Jacobian $J_X$ of $X$, $\hvarpi_J: \hJ_X \to J_X$, and $\hkappa_J : \CC^3 \to \hJ_X:= \CC^3/(\Gamma_\hX\cap \CC^3)$.

Since for each branch point $B_j:=(b_j, 0)\in X$ $(j=1, \ldots, 6)$, we have double branch points $\hB^\pm_j:=(\pm\sqrt{e_j},0) \in \hX$ as illustrated in Figure \ref{fg:Fig00}.

Similar to the Jacobi elliptic functions, $\hJ_X=\CC^3/\hGamma_X$ is determined by the same Abelian integral $\tv$, and thus we use the same symbol $\tv$ as $\tv : S^3\tX \to \CC^3$ for $\hX$ \cite{MP15}.

\begin{figure}
\begin{center}

\includegraphics[width=0.6\hsize]{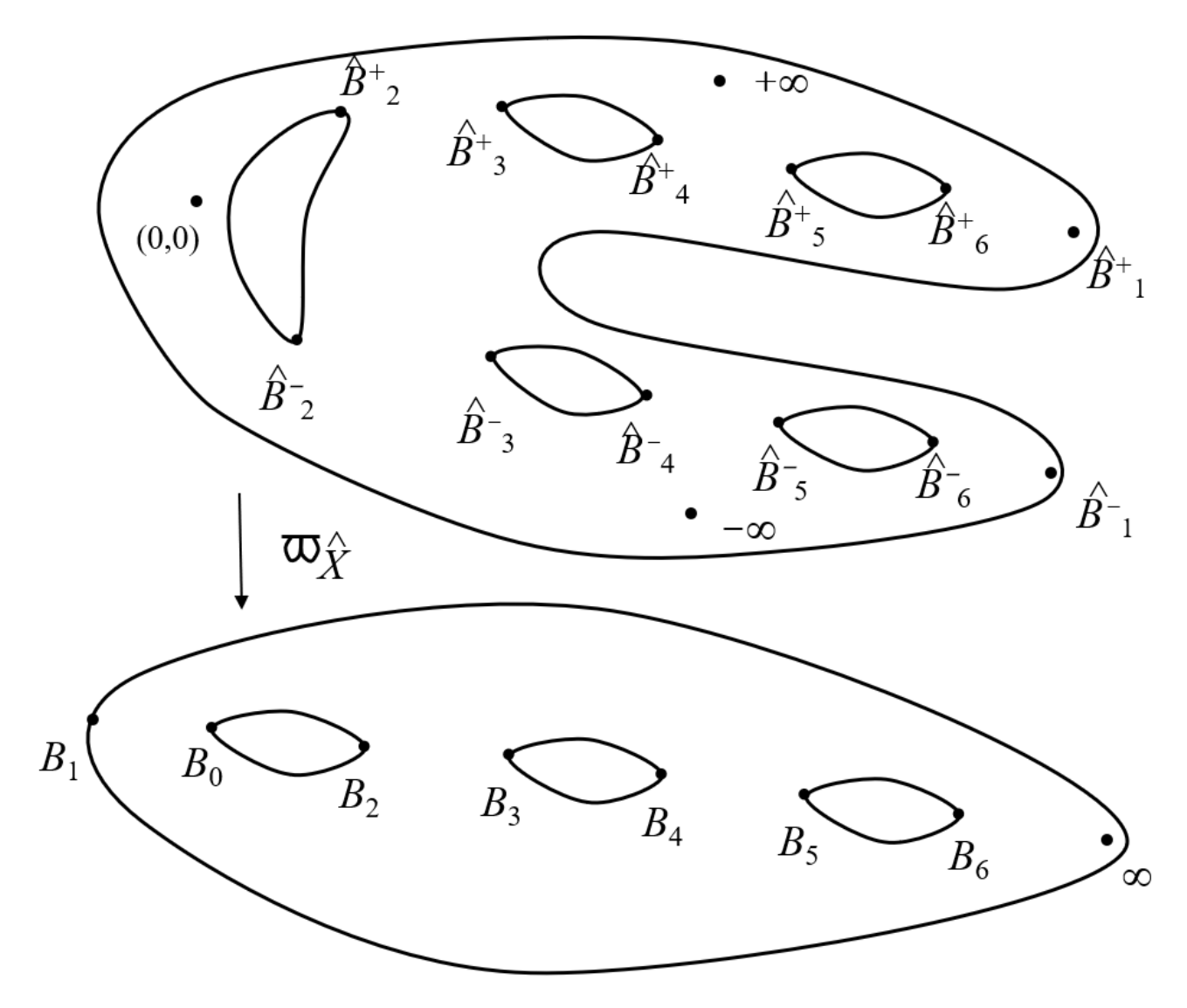}

\end{center}

\caption{The double covering $\varpi_\hX: \hX \to X$,
$\varpi_X: (w, z) \mapsto (w^2+b_0, zw)=(x,y)$.
}\label{fg:Fig00}
\end{figure}

\cite{Mat02b} shows the hyperelliptic solutions of the MKdV equation over $\CC$.

\begin{proposition} {\textrm{\cite{Mat02b}}}
\label{4th:MKdVloop}
For $((x_1,y_1),\cdots, (x_3,y_3)) \in S^3 X$, the fixed branch point $(b_0, 0)$, and $u:= v( (x_1,y_1),$ $\cdots,(x_3,y_3))$,
$$
\displaystyle{
   \psi(u) :=
-\ii \log (b_0-x_1)(b_0-x_2)(b_0-x_3)
}
$$
satisfies the MKdV equation over $\CC$,
\begin{equation}
(\partial_{u_{2}}-\frac{1}{2}
(\lambda_{6}+3b_0)
          \partial_{u_{3}})\psi
           +\frac{1}{8}
\left(\partial_{u_3} \psi\right)^3
 +\frac{1}{4}\partial_{u_3}^3 \psi=0,
\label{4eq:loopMKdV2}
\end{equation}
where $\partial_{u_i}:= \partial/\partial u_i$ as an differential identity in $S^3 X$ and $\CC^3$.
\end{proposition}

We, here, emphasize the difference between the focusing MKdV equations (\ref{4eq:MKdV1phi}) over $\RR$ and (\ref{4eq:loopMKdV2}) over $\CC$.
In (\ref{4eq:MKdV1phi}), $\phi$ is a real valued function over $\RR^2$ but $\psi$ in (\ref{4eq:loopMKdV2}) is a complex valued function over $\CC^2 \subset \CC^3$.
As mentioned in Introduction, we describe the difference.
By introducing real and imaginary parts, $ u_a = u_{a\,\rr} + \ii u_{a\,\ri}$, $(a=1,2,3)$, and $ \psi = \psi_{\rr} + \ii \psi_{\ri}$, the real part of (\ref{4eq:loopMKdV2}) is reduced to the focusing gauged MKdV (FGMKdV) equation with the gauge field $A(u)=(\lambda_{2g}+3b_0+\frac{3}{4}(\partial_{u_{3}\, \rr}\psi_\ri)^2)/2$,
\begin{equation}
(\partial_{u_{2}\, \rr}-
A(u)\partial_{u_{3}\, \rr})\psi_\rr
           +\frac{1}{8}
\left(\partial_{u_3\, \rr} \psi_\rr\right)^3
+\frac{1}{4}\partial_{u_3\, \rr}^3 \psi_\rr=0
\label{4eq:gaugedMKdV2}
\end{equation}
and its associated equation for $\tA(u)=(\lambda_{2g}+3b_0-\frac{3}{4}(\partial_{u_{3}\, \rr}\psi_\rr)^2)/2$,
\begin{equation}
(\partial_{u_{2}\, \rr}-
\tA(u)\partial_{u_{3}\, \rr})\psi_\ri
           -\frac{1}{8}
\left(\partial_{u_3\, \rr} \psi_\ri\right)^3
+\frac{1}{4}\partial_{u_3\, \rr}^3 \psi_\ri=0
\label{4eq:gaugedMKdV2i}
\end{equation}
by the Cauchy-Riemann relations as mentioned in \cite[(11)]{MP22}.
We call (\ref{4eq:gaugedMKdV2}) and (\ref{4eq:gaugedMKdV2i}) the coupled FGMKdV equations.
We note that $\partial_{u_a} \psi = \partial_{u_a\, \ri} (\psi_{\rr} + \ii \psi_{\ri})$ because $\partial_{u_a} =(\partial_{u_a\, \rr}-\ii\partial_{u_a\, \ri})/2$, and thus $(\partial_{u_3} \psi)^3$ contains the term $-3(\partial_{u_3\,\rr}\psi_{\ri})^2 \partial_{u_3\,\rr}\psi_{\rr}$.

\bigskip
A solution of (\ref{4eq:MKdV1phi}) in terms of the data in Proposition \ref{4th:MKdVloop} must satisfy  the following conditions \cite{MP22}:

\begin{condition}\label{cnd}
{\rm{
\begin{enumerate}

\item[CI] $\prod_{i=1}^3 |x_i - b_0|=$ a constant $(> 0)$ in Proposition \ref{4th:MKdVloop},

\item[CII] $d u_{3\,\ri}=d u_{2\, \ri}=0$ or $d u_{3\, \rr}=d u_{2\, \rr}=0$ in Proposition \ref{4th:MKdVloop}, and

\item[CIII] $A(u)$ is a real constant:
if $A(u)=$ constant (or $\partial_{u_{3}\, \rr}\psi_\ri=$ constant), (\ref{4eq:gaugedMKdV2}) is reduced to (\ref{4eq:MKdV1phi}), i.e., $\psi_\rr=\phi$. 
\end{enumerate}
}}
\end{condition}
It is obvious that if we have the solutions $\psi_\rr$ of (\ref{4eq:gaugedMKdV2}) satisfying the conditions CI--CIII, $\partial_{u_3, \rr}\psi_{\rr}/2$ obeys the focusing MKdV equation (\ref{4eq:MKdV1phi}), though there is a simple shift of $\partial_t $ to $\partial_t + c \partial_s$ for a certain real constant $c$.

However, in this paper we focus on the conditions CI and CII and the real hyperelliptic solutions of the FGMKdV equation (\ref{4eq:gaugedMKdV2}) of genus three.

\section{Hyperelliptic curves of genus three in angle expression}\label{sec:g=3}
We will review the angle expression of hyperelliptic curves based on  \cite{Ma24b} for the sake of self-containment in this paper.
This section corresponds to Section 4 in \cite{Ma24b} except Figure \ref{fg:Fig01} and Assumption \ref{Asmp}.
The difference causes that the region in $\hX$ that we treat in this paper differs from that of the previous paper \cite{Ma24b}.
Since the region consists of disjoint arcs, we will handle the configurations of points in $\hX$ of $S^3\hX$ more easily than in the previous paper \cite{Ma24b}.

To find real solutions of the FGMKdV equation (\ref{4eq:gaugedMKdV2}) under Condition \ref{cnd}  CI and CII, we introduce the angle expression \cite{Mat07, MP22, Ma24b} as mentioned in Introduction.

We restrict the moduli (rather, parameter) space of the curve $X$ by the following.
We choose coordinates $u = {}^t(u_1, u_2, u_3)$
 in  $\CC^3$;
$u_i = u_i^{(1)}+u_i^{(2)}+u_i^{(3)}$, where $u_i^{(j)}
=v_i((x_j, y_j))$ for $(x_j, y_j) \in X$.
There are the projection $\varpi_x : X \to \PP^1$, $((x,y) \mapsto x)$, and similarly $\hvarpi_x : \hX \to \PP^1$, $((w,z) \mapsto w)$.

We let $b_0=-1$ and $e_j := b_j - b_0$ $(j=1,2,\ldots,6)$ satisfying the following relations,
$$\sqrt{e_{2a-1}} = \alpha_a +\ii \beta_a,
\quad
\sqrt{e_{2a}} = \alpha_a -\ii \beta_a,
$$
where $\alpha_a, \beta_a\in \RR$, $a =1,2,3$, satisfying 
$\alpha_a^2 + \beta_a^2 = 1$.

We recall $w^2 = (x-b_0)$ and $w=\ee^{\ii \varphi}$.
For the real expression of (\ref{4eq:hypC}), we use the generalization of the angle expression of the elliptic integral as mentioned in \cite{Mat07, MP22, Ma24b}.
Under these assumptions, we have the real extension of the hyperelliptic curve $X$ by $(\ee^{\ii\varphi}, y/\ee^{\ii\varphi})\in \hX$. 
The direct computation shows the following:

\begin{lemma} \label{4lm:g3gene_y2}
Let $ \ee^{2\ii\varphi} :=(x-b_0)$,
(\ref{4eq:hypC}) is written by
\begin{gather}
y^2=-64 \frac{4\ee^{8\ii\varphi}}{k_1^2 k_2^2k_3^2} 
(1-k_1^2 \sin^2 \varphi)(1-k_2^2 \sin^2 \varphi)
(1-k_3^2 \sin^2 \varphi),
\label{4eq:HEcurve_phi}
\end{gather}
where 
$\displaystyle{
k_a = \frac{2\ii\sqrt[4]{e_{2a-1}e_{2a}}}{\sqrt{e_{2a-1}}- \sqrt{e_{2a}}}
=\frac{1}{\beta_a}}$, $(a=1,2,3)$.
\end{lemma}

\begin{proof}
See Lemma 4.2 in \cite{Ma24b}.
\end{proof}

We are interested in the case that $y/\ee^{4\ii \varphi}$ of the part of the holomorphic one-form is real or pure imaginary.
We assume that the branch points surround the circle whose center is $(b_0,0)$ and radius is $1$.
As mentioned in Section 2, we handle the double covering $\hX$ with twelve branch points.
We define $\varphi_{\fb a}^{+\pm}:=\pm \sin^{-1}(1/k_a)$ and $\varphi_{\fb a}^{-\pm}
=\pi - \varphi_{\fb a}^{+\pm}$ $(a= 1, 2, 3)$ as in Figure \ref{fg:Fig01}.
In other words, we have twelve branch points $(\ee^{\varphi_{\fb a}^{\pm\pm}}, 0)$ in $\hX$, $a=1, 2, 3$;
$\{(\ee^{\varphi_{\fb a}^{\pm\pm}}, 0)\}$ corresponds to the branch points $\{\hB_{b}^\pm\}$ in Figure \ref{fg:Fig00}, although the correspondences may look unclear due to the angle expression $\varphi$. 
Thus, we use the symbol $\fb$ for these branch points.

\begin{assumption}\label{Asmp}
{\rm{
As in Figure \ref{fg:Fig01}, we assume that $k_1> k_2 > k_3>1.0$ and 
$$
\varphi_{\fb}^{[1\pm]}:=\varphi_{\fb1}^{+\pm}, \quad
\varphi_{\fb}^{[2+]}:=\varphi_{\fb3}^{++}, \quad
\varphi_{\fb}^{[2-]}:=\varphi_{\fb2}^{++},\quad
\varphi_{\fb}^{[3+]}:=\varphi_{\fb2}^{+-}, \quad
\varphi_{\fb}^{[3-]}:=\varphi_{\fb3}^{+-}.
$$
}}
\end{assumption}

\bigskip

Further, by using (\ref{2eq:al_hyp_kappas}) we may introduce a continuous transcendent map for a certain open subspace $\tU$ of $\tX$,
\begin{equation}
\Phi: \tX \supset \tU \to \CC, \quad\left( \gamma \mapsto =\frac{1}{2\ii}\log (\hvarpi_x\hkappa_X\gamma -b_0)\right)
\label{eq:trans_Phi}
\end{equation}
under Assumption \ref{Asmp}.
The $\varphi$ is defined for the neighborhood of the circle at the origin of the curve $\hX$.
The multiplicity of the logarithm function, if exists, may correspond to the multiplicity of the covering $\hkappa_X : \tX \to \hX$.
Thus, we will only use this map for a limited region as follows.

\begin{figure}
\begin{center}

\includegraphics[width=0.42\hsize, bb= 0 0 641 668]{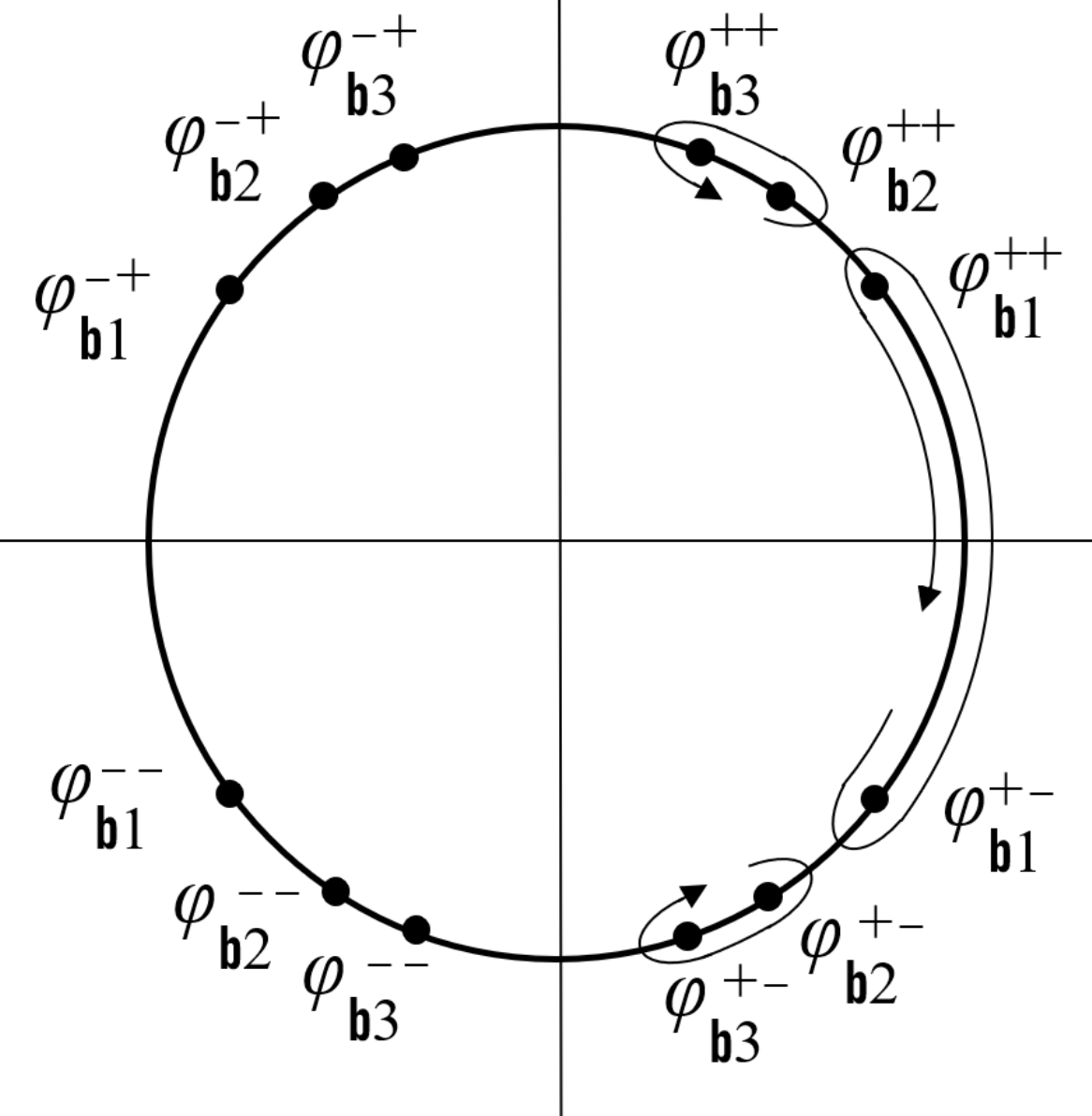}

\end{center}

\caption{
The orbits of each $\varphi_i$ in the quadrature:
$k_1> k_2 > k_3>1.0$.
}\label{fg:Fig01}
\end{figure}

\bigskip

We consider a point $((x_1, y_1),\ldots,(x_3,y_3))$ in $S^3 X$ under the condition CI, $|x_j-b_0|=1$, ($b_0 = -1$).
We define the variable $\varphi_j$ by $x_j=  \ee^{\ii \varphi_j}
(\ee^{\ii \varphi_j}+ \ee^{-\ii \varphi_j})$ $=2 \ii  \ee^{\ii \varphi_j} \sin \varphi_j$, $(j=1,2,3)$.
Noting
$d x_j = 2 \ii  \ee^{2\ii \varphi_j}d\varphi_j$ and
 $x_j^\ell d x_j = (2\ii)^{\ell+1}  \ee^{(2+\ell)\ii \varphi_j}\sin^\ell \varphi_j\ d \varphi_j$,
we have the holomorphic one forms
$(\nuI{1}^{(j)}, \nuI{2}^{(j)},  \nuI{3}^{(j)})$ $(j=1,2,3)$, 
\begin{equation}
\left(
\frac{ \ee^{-2\ii\varphi_j}\ d \varphi_j}{
8  K(\varphi_j)},
\frac{-\ii \ee^{-\ii\varphi_j}\sin(\varphi_j)\ d \varphi_j}{
4  K(\varphi_j)},
\frac{ -\sin^2 \varphi_j\ d \varphi_j}{
2K(\varphi_j)}\right),
\label{eq:3.3_1stdiff}
\end{equation}
where $K(\varphi):=\tgamma\tK(\varphi)$, $\displaystyle{
\tK(\varphi):=\frac{\sqrt{
 (1-k_1^2 \sin^2 \varphi)(1-k_2^2 \sin^2 \varphi)
(1-k_3^2 \sin^2 \varphi)}}
{k_1k_2k_3}}$, and $\tgamma = \pm 1$.
$\tgamma$ represents the sheet of $\hvarpi_x:\hX \to \PP^1$, \break
$(( \ee^{\ii \varphi},
\ii 8\tgamma \tK(\varphi)\ee^{3\ii \varphi}) \mapsto \ee^{\ii \varphi})$.

Using the ambiguity $\tgamma$, we handle $-(\nuI{1}^{(j)}, \nuI{2}^{(j)},  \nuI{3}^{(j)})$ rather than \break $(\nuI{1}^{(j)}, \nuI{2}^{(j)},  \nuI{3}^{(j)})$ $(j=1,2,3)$ from here.

Then we obviously have the following lemmas following \cite{Ma24b} (Lemmas 4.4 and 4.5 in \cite{Ma24b}):

\begin{lemma} \label{4lm:dudphi}
For $( \ee^{\ii \varphi_j}, \ii8 K_j\ee^{3\ii\varphi_j})_{j=1, 2, 3}\in S^3\hX$, where $K_j:=\tgamma_j\tK(\varphi_j)$, ($j=1, 2, 3$), the following holds:

\noindent
$$\displaystyle{
\begin{pmatrix} d u_1 \\ d u_2\\ du_3\end{pmatrix}
=-
\begin{pmatrix}
\frac{ \ee^{-2\ii\varphi_1}}{8 K_1}&
\frac{ \ee^{-2\ii\varphi_2}}{8 K_2}&
\frac{ \ee^{-2\ii\varphi_3}}{8 K_3}\\
\frac{\ii \ee^{-\ii\varphi_1}\sin(\varphi_1)}{4 K_1}&
\frac{\ii \ee^{-\ii\varphi_2}\sin(\varphi_2)}{4 K_2}&
\frac{\ii \ee^{-\ii\varphi_3}\sin(\varphi_3)}{4 K_3}\\
\frac{- \sin^2(\varphi_1)}{2 K_1}&
\frac{- \sin^2(\varphi_2)}{2 K_2}&
\frac{- \sin^2(\varphi_3)}{2 K_3}\\
\end{pmatrix}
\begin{pmatrix} d \varphi_1 \\ d \varphi_2 \\d \varphi_3\end{pmatrix}
}.
$$
Let the matrix be denoted by $\cL$.
Then the determinant of $\cL$, 
$$
\det(\cL)=\displaystyle{
\frac{\sin(\varphi_1-\varphi_2)\sin(\varphi_2-\varphi_3)\sin(\varphi_3-\varphi_1)}{4^3 K_1 K_2 K_3}}.
$$
\end{lemma}

We also have the inverse of Lemma \ref{4lm:dudphi} at a regular locus:

\begin{lemma} \label{4lm:dudphi3}
For $\varphi_a \in [\varphi_\fb^{[a-]}, \varphi_\fb^{[a+]}]$, $(a=1,2,3)$ such that $\varphi_a \neq \varphi_b$ $(a\neq b)$, we have
\begin{equation}
\displaystyle{
\begin{pmatrix} d \varphi_1 \\ d \varphi_2 \\ d \varphi_3\end{pmatrix}
=\cK \cM
\begin{pmatrix} d u_1 \\ d u_2\\ du_3\end{pmatrix}
},
\qquad
\cL^{-1}=\cK \cM,
\label{4eq:Elas3.4}
\end{equation}
where
$
\displaystyle{
\cK
:=
-\begin{pmatrix}
\frac{K_1}{\sin(\varphi_2-\varphi_1)\sin(\varphi_3-\varphi_1)}&0& 0 \\
0& \frac{K_2}{\sin(\varphi_3-\varphi_2)\sin(\varphi_1-\varphi_2)}&0 \\
0&0&\frac{K_3}{\sin(\varphi_1-\varphi_3)\sin(\varphi_2-\varphi_3)} 
\end{pmatrix}
}
$

\noindent
and,
{\small{
$$
%\displaystyle{
\cM
:=
\begin{pmatrix}
8  \sin\varphi_2\sin\varphi_3&
4(2 \sin\varphi_2\sin\varphi_3 +\ii \sin(\varphi_2+\varphi_3) )&
-2 \ee^{-\ii (\varphi_2+\varphi_3)} \\
8  \sin\varphi_1\sin\varphi_3&
4(2  \sin\varphi_1\sin\varphi_3 +\ii \sin(\varphi_3+\varphi_1) )&
-2 \ee^{-\ii (\varphi_1+\varphi_3)} \\
8  \sin\varphi_1\sin\varphi_2&
4(2  \sin\varphi_1\sin\varphi_2 +\ii \sin(\varphi_1+\varphi_2) )&
-2\ee^{-\ii (\varphi_1+\varphi_2)} \\
\end{pmatrix}.
$$
}}
\end{lemma}

\begin{proof}
See Lemma 4.5 in \cite{Ma24b}. \qed
\end{proof}

We remark that (\ref{4eq:Elas3.4}) in Lemma \ref{4lm:dudphi3} means that even if $\varphi_i$ $(i=1,2,3)$ is real, $d\varphi_j$ is complex valued one-form.
We let it decomposed to $d\varphi_j = d\varphi_{j,\rr}+ \ii d \varphi_{j, \ri}$.
Further, we introduce $\varphi := \varphi_1 +\varphi_2 +\varphi_3 \in \RR$ and $d\varphi = d\varphi_{\rr}+ \ii d \varphi_{\ri}$; 
$\psi_\rr = 2 \varphi$, $d\psi_\rr = 2 d\varphi_\rr$ and $d\psi_\ri =2 d\varphi_\ri$ for $\psi$ in (\ref{4eq:loopMKdV2}) and (\ref{4eq:gaugedMKdV2}).
We sometimes write $\varphi_{a,\rr}:=\varphi_a$.

\begin{remark}
{\rm{
For a point $\gamma' \in \tX$, the holomorphic one form $\nu(\gamma')$ is regarded as $\nu(\gamma')=\nu(\hkappa_X \gamma')$.
Lemma \ref{4lm:dudphi} means that for $u=\tv(\gamma=(\gamma_1,\gamma_2, \gamma_3)) = \displaystyle{\sum \int_{\gamma_i} \nu}$, $du = d\tv(\gamma)$ is equal to $\sum \nu(\gamma_i)$.
We restrict $\tU \subset S^3 \tX$ the domain of $\Phi$ in (\ref{eq:trans_Phi}) and $\gamma =(\gamma_1, \gamma_2, \gamma_3)\in \tU$ such that $\prod_{a=1}^3 [\varphi_\fb^{[a-]}, \varphi_\fb^{[a+]}]\subset \Phi\circ\kappa_X S^3 \tU$.
We regard this as a linear transformation between $du$ and $d \gamma_i$, i.e., $du = \cL(d \Phi(\gamma))$, where $d\varphi=d \Phi(\gamma)= \displaystyle{\begin{pmatrix} d\varphi_1\\ d\varphi_2\\ d\varphi_3\end{pmatrix}}$.
The matrix $\cL$ is the matrix representation of the transformation $\cL$ $:T^*_{\Phi(\gamma)} \Phi(S^3 \tU) \to T^*_u\CC^3$.
Since $\cL$ is invertible due to the Abel-Jacobi theorem for the regular locus, $d \varphi= \cL^{-1}(du)= \cK \cM du$ there.
Then the transformation $\cL^{-1}$ (or the matrix $\cK \cM$) can be also interpreted as the pullback  $\tv^* : T^*_{ u} \CC^3 \to T^*_{\gamma} S^3 \tX$ for a point $\tv(\gamma)=u \in \CC^3$ for the Abelian integral $\tv: S^3 \tX \to \CC^3$, and for the Abel-Jacobi map $\hv: S^3 \hX \to \hJ_X$.
Then we regard $\cL$ as $\tv^{-1*}: T^*_{\gamma} S^3 \tX \to T^*_u\CC^3$ up to $\Phi$.

We note that since $d \varphi_i$ is defined on an open set in $\CC$ but $\nu(\gamma_i)$ is defined on $\hX$ related to $\tU$ in (\ref{eq:trans_Phi}), the sign $\tgamma_i$ in $K_j$ is defined for the lift from the open set of $\CC$ to $\hX$.
Of course, the lift is not a map but we consider $\gamma_i \in \tX$ and thus the sign $\tgamma_i$ is uniquely determined for a given $\gamma_i\in \tX$.

Further, we handle these correspondences only for the regular locus in this paper.
}}
\end{remark}

\section{Real hyperelliptic solutions of the gauged MKdV equation over $\RR$}

Although this section corresponds to Section 5 and Section 6 in the previous paper \cite{Ma24b}, we will modify the quantities introduced in the previous paper to find a nicer parameterization than \cite{Ma24b} to obtain Theorem \ref{th:4.2}.
The treatments are very similar to \cite{Ma24b}, but slightly different definitions of the quantities lead us to have a nice property in Lemma \ref{lm:4.7}.
We finally obtain the real hyperelliptic solutions of the FGMKdV equation as in Theorem \ref{th:4.2}.

Since we employ the region that we handle as illustrated in  Figure \ref{fg:Fig01} and Assumption \ref{Asmp}, the global behavior of the differential equation (\ref{eq:g3CIII}) is simply obtained as in Theorem \ref{pr:solgMKdV} even for the hyperelliptic Riemann surface $\hX$.

Thus, this section is a key section in this paper, which allows us to consider the numerical evaluation of the FGMKdV equation and a generalization of Euler's elastica in Section 5.

\bigskip

We will go on to assume the non-singular locus of $\tv^*$ and $\tv^{-1*}$.
We recall $\tU$ as the domain of $\Phi$ in (\ref{eq:trans_Phi}).
For a point $\gamma \in S^3\tU$ and $u=\tv(\gamma) \in \CC^3$, the fibers $T^*_{ \gamma}S^3 \tX$ (or $T^*_{\Phi(\gamma)} S^3 \tU$) and $T^*_{u} \CC^3$ are isomorphic to $\CC^3$ as the real or the complex vector space.
When we regard the matrix $\cK \cM$ as a linear map from $\CC^3$ to $\CC^3$ or the pullback $\tv^*:T^*_{u} \CC^3 \to T^*_\gamma S^3 \tX$  for the point $\gamma \in S^3 \tX$, we decompose the image of the map $\tv^*$ and $\tv^{-1*}$ from the viewpoint of the reality conditions CI and CII so that the decomposition induces a map from $\RR^\ell$ to $\RR^\ell$.

Direct computations show the following lemma:
\begin{lemma}\label{lm4.1}
For the real meromorphic valued vectors on $S^3 \hX$,
\begin{equation*}
\cV_1
:=-\begin{pmatrix}
\frac{8\tK_1\sin(\varphi_2)\sin(\varphi_3)}{\sin(\varphi_2-\varphi_1)\sin(\varphi_3-\varphi_1)}\\
\frac{8\tK_2\sin(\varphi_3)\sin(\varphi_1)}{\sin(\varphi_3-\varphi_2)\sin(\varphi_1-\varphi_2)}\\
\frac{8\tK_3\sin(\varphi_1)\sin(\varphi_2)}{\sin(\varphi_1-\varphi_3)\sin(\varphi_2-\varphi_3)}
\end{pmatrix},\quad
\cV_2
:=-\begin{pmatrix}
\frac{2\tK_1\sin(\varphi_2+\varphi_3)}{\sin(\varphi_2-\varphi_1)\sin(\varphi_3-\varphi_1)}\\
\frac{2\tK_2\sin(\varphi_3+\varphi_1)}{\sin(\varphi_3-\varphi_2)\sin(\varphi_1-\varphi_2)}\\
\frac{2\tK_3\sin(\varphi_1+\varphi_2)}{\sin(\varphi_1-\varphi_3)\sin(\varphi_2-\varphi_3)}
\end{pmatrix},
\end{equation*}
\begin{equation}
\cV_3
:=\begin{pmatrix}
\frac{2\tK_1\cos(\varphi_2+\varphi_3)}{\sin(\varphi_2-\varphi_1)\sin(\varphi_3-\varphi_1)}\\
\frac{2\tK_2\cos(\varphi_3+\varphi_1)}{\sin(\varphi_3-\varphi_2)\sin(\varphi_1-\varphi_2)}\\
\frac{2\tK_3\cos(\varphi_1+\varphi_2)}{\sin(\varphi_1-\varphi_3)\sin(\varphi_2-\varphi_3)}
\end{pmatrix},\quad
\cC=\displaystyle{
\begin{pmatrix} \tgamma_1 & & \\ &\tgamma_2 & \\ & & \tgamma_3\end{pmatrix}},
\label{eq:g3CIII0}
\end{equation}
where $\tK_j:=\tK(\varphi_j)$, 
we have the following relations,
\begin{equation}
\cK\cM= (\cV_1, \cV_1+2\ii \cV_2, \cV_3+\ii\cV_2)
= (\cV_1, \cV_2, \cV_3)\begin{pmatrix} 1 & 1 & 0\\ 0& 2\ii & \ii \\ 0 & 0 & 1
\end{pmatrix}.
\label{eq:g3KM=V}
\end{equation}
\end{lemma}

\begin{proof}
We note $\cC^{-1}=\cC$. 
The direct computation show them.
\end{proof}

Note that these $\cV_a$ differ slightly from those in Lemma 5.1 in \cite{Ma24b}, and thus the following also differs from the previous paper \cite{Ma24b}.

Since in terms of $\cV$'s, $\cK\cM$, i.e., 
$\begin{pmatrix} d\varphi_1 \\ d \varphi_2 \\ d \varphi_3\end{pmatrix}
=\cK\cM\begin{pmatrix} du_1 \\ d u_2 \\ d u_3\end{pmatrix}$ is expressed by (\ref{eq:g3KM=V}) in Lemma \ref{lm4.1}, we may introduce new parameters $t_a$ in $\CC^3$.

\begin{lemma}
By letting
\begin{equation}
\cB:=\begin{pmatrix} 1 & 1 & 0\\ 0& 2\ii & \ii \\ 0 & 0 & 1
\end{pmatrix}, \quad
\begin{pmatrix} dt_1 \\  d t_2 \\ d t_3\end{pmatrix}
:=\cB
\begin{pmatrix} du_1 \\ d u_2 \\ d u_3\end{pmatrix}, 
\label{eq:dus_dts}
\end{equation}
we have 
\begin{equation}
\begin{pmatrix} d\varphi_1 \\ d \varphi_2 \\ d \varphi_3\end{pmatrix}
=(\cV_1,  \cV_2, \cV_3)\begin{pmatrix} dt_1 \\ d t_2 \\ d t_3\end{pmatrix},
\quad
\begin{pmatrix} du_1 \\ d u_2 \\ d u_3\end{pmatrix}
=\begin{pmatrix} 1 & -\ii/2 & 1/2\\ 0 & \ii/2 & -1/2\\ 0 & 0 & 1
\end{pmatrix}
\begin{pmatrix} dt_1 \\  d t_2 \\ d t_3\end{pmatrix}.
\label{eq:dps_dts}
\end{equation}
\end{lemma}

We let the matrix in (\ref{eq:dps_dts}) denoted by $\cB^{-1}$.
We have the following lemma:
\begin{lemma}
Let the basis of $\CC^3$ be
$$
\displaystyle{\left\{ 
\be_1:=\begin{pmatrix} 1 \\ 0 \\ 0\end{pmatrix}, 
\be_2:=\begin{pmatrix}\ii/2 \\ -\ii/2 \\ 0\end{pmatrix},
\be_3:=\begin{pmatrix} 1/2 \\ -1/2 \\ 1\end{pmatrix}
\right\}},
$$
i.e.,
\begin{equation}
\CC^3 = \langle\be_1, \be_2, \be_3\rangle_\CC, \quad
\RR^2 = \langle\be_1, \be_3\rangle_\RR \subset \CC^3.
\label{eq:basis_in_C3}
\end{equation}
Then we have
$$
(\be_1, \be_2, \be_2)
= \cL \cC (\cV_1, \cV_2, \cV_3).
$$
\end{lemma}

Since (\ref{eq:dps_dts}) also shows
{\small{
$$
(\cV_1,  \cV_2, \cV_3)=
\begin{pmatrix}
\partial_{t_1} \varphi_1&   \partial_{t_2} \varphi_1& \partial_{t_3} \varphi_1\\
\partial_{t_1} \varphi_2&  \partial_{t_2} \varphi_2& \partial_{t_3} \varphi_2\\
\partial_{t_1} \varphi_3&   \partial_{t_2} \varphi_3& \partial_{t_3} \varphi_3
\end{pmatrix}, \quad
\begin{pmatrix} d\varphi_1 \\ d \varphi_2 \\ d \varphi_3\end{pmatrix}
=
\begin{pmatrix}
\partial_{t_1} \varphi_1&   \partial_{t_2} \varphi_1& \partial_{t_3} \varphi_1\\
\partial_{t_1} \varphi_2&   \partial_{t_2} \varphi_2& \partial_{t_3} \varphi_2\\
\partial_{t_1} \varphi_3&   \partial_{t_2} \varphi_3& \partial_{t_3} \varphi_3
\end{pmatrix}
\begin{pmatrix} dt_1 \\ d t_2 \\ d t_3\end{pmatrix},
$$
}}
the relations between differential operators are given by
{\small{
\begin{equation}
\begin{pmatrix} 
\partial_{t_1} \\ 
\partial_{t_2} \\
\partial_{t_3} \end{pmatrix}
=\begin{pmatrix} 1 & 0 & 0\\ \ii/2& -\ii/2 & 0\\ 1/2 & -1/2 & 1
\end{pmatrix}
\begin{pmatrix} 
\partial_{u_1} \\ \partial_{u_2} \\ \partial_{u_3}\end{pmatrix},
\quad
\begin{pmatrix} 
\partial_{u_1} \\ \partial_{u_2} \\ \partial_{u_3}\end{pmatrix}
=\begin{pmatrix} 1 & 0 & 0\\ 1& 2\ii & 0 \\ 0 & \ii & 1
\end{pmatrix}
\begin{pmatrix} 
\partial_{t_1} \\ 
\partial_{t_2} \\
\partial_{t_3} \end{pmatrix}.
\label{eq:pus_pts}
\end{equation}
}}
We note that they are given by $\trp\cB^{-1}$ and $\trp\cB$. 
Then we have the following the relation lemma:
\begin{lemma}\label{lm:4.3a}
$$
\partial_{t_a}\psi = (1, 1, 1)\cV_a, \quad (a=1, 2, 3),
$$
$$
\partial_{u_1}\psi = (1, 1, 1)\cV_1, \quad
\partial_{u_2}\psi = (1, 1, 1)(\cV_1+2\ii \cV_2),\quad
$$
$$
\partial_{u_3}\psi = (1, 1, 1)(\cV_3+\ii \cV_2).
$$
\end{lemma}

Here we recall the Cauchy-Riemann relations of these parameters:
\begin{proposition}\label{pr:4.4a}
Let $u_a = u_{a\rr}+\ii u_{a\ri}$ and $t_a = t_{a\rr}+\ii t_{a\ri}$ for $a=1, 2, 3$.
For a complex analytic function $\psi = \psi_\rr + \ii \psi_\ri$,
(\ref{eq:pus_pts}) shows 
$$
\partial_{u_a\rr}\psi_\rr=\partial_{u_a\ri}\psi_\ri , \qquad
\partial_{u_a\rr}\psi_\ri=-\partial_{u_a\ri}\psi_\rr.
\qquad (a= 1, 2, 3),
$$
$$
\partial_{t_a\rr}\psi_\rr=\partial_{t_a\ri}\psi_\ri , \qquad
\partial_{t_a\rr}\psi_\ri=-\partial_{t_a\ri}\psi_\rr.
\qquad (a= 1, 2, 3).
$$
\end{proposition}
\begin{proof}
Since the Cauchy-Riemann relations are given by $\overline{\partial_{u_a}}\psi =\overline{\partial_{t_a}}\psi=0$, we obtain them.
\end{proof}

Lemma \ref{lm:4.3a} and Proposition \ref{pr:4.4a} show the following lemma.

\begin{lemma}
For $\varphi_a\in \RR$, i.e., $\varphi_{a,\ri}=0$, $(a=1, 2, 3)$,
$$
\partial_{t_a\rr }\psi_{\ri} = \partial_{t_a \ri} \psi_\rr = 0, \quad (a=1,2,3).
$$
\end{lemma}

It implies that $t_a \in \RR$ belongs to $\CC^3$.
Since we are concerned with the case that $\varphi_a\in \RR$, i.e., $\varphi_{a,\ri}=0$, $(a=1, 2, 3)$.

\begin{lemma}\label{lm:4.7}
For $\varphi_a\in \RR$, i.e., $\varphi_{a,\ri}=0$, $(a=1, 2, 3)$,
the following relations hold:
\begin{enumerate}

\item $\partial_{u_1} \psi = \partial_{u_1,\rr} \psi_\rr= \partial_{u_1,\rr} \psi_\rr
= \partial_{t_1,\rr} \psi_\rr
= (1,1,1)\cV_1$, $\partial_{u_1\rr }\psi_{\ri} = \partial_{u_1 \ri} \psi_\rr = 0$.

\item 
$\partial_{u_2} \psi = \partial_{u_2,\rr} \psi_\rr-\ii\partial_{u_2,\ri} \psi_\rr= \partial_{u_2,\rr} \psi_\rr+\ii\partial_{u_2,\rr} \psi_\ri$

$= \partial_{t_1,\rr} \psi_\rr+2\ii\partial_{t_2,\rr} \psi_\rr
= (1,1,1)(\cV_3 + 2\ii \cV_2)$.

\item
$
\partial_{u_3} \psi = \partial_{u_3,\rr} \psi_\rr-\ii\partial_{u_3,\ri} \psi_\rr= \partial_{u_3,\rr} \psi_\rr+\ii\partial_{u_3,\rr} \psi_\ri$

$= \partial_{t_3,\rr} \psi_\rr+\ii\partial_{t_2,\rr} \psi_\rr= (1,1,1)(\cV_3 + \ii \cV_2)$.

\item
Particularly we have
\begin{eqnarray}
&\partial_{u_2,\rr} \psi_\rr=\partial_{t_1,\rr} \psi_\rr=(1,1,1)\cV_1, \quad
&\partial_{u_2,\rr} \psi_\ri=2\partial_{t_2,\ri} \psi_\ri=2(1,1,1)\cV_2,
\nonumber \\
&\partial_{u_3,\rr} \psi_\rr=\partial_{t_3,\rr} \psi_\rr=(1,1,1)\cV_3, \quad
&\partial_{u_3,\rr} \psi_\ri=\partial_{t_2,\ri} \psi_\ri=(1,1,1)\cV_2.
\end{eqnarray}
\end{enumerate}
\end{lemma}
\begin{proof}
We obviously obtain them.
\end{proof}

\begin{remark}
{\rm{
We note the following.
Lemma \ref{lm:4.7} does not claim that there is a different complex structure in $J_X$ and the image of $\tv$.
These parameterizations are consistent only for the local regions related to the arcs of $S^1$ in $\hkappa_X \hX$, or $\varphi_a \in \RR$ and $\varphi_{a,\ri}=0$ ($a=1,2,3$).
This means that we simply embed the real vector space $\RR^3$ in $\CC^3$ via the matrix $\cB$ and $\cB^{-1}$.
}}
\end{remark}

\bigskip

The coupled FGMKdV equations (\ref{4eq:gaugedMKdV2}) with (\ref{4eq:gaugedMKdV2i}) can be expressed in terms of the parameterizations of $t$.
We give the first theorem in this paper.
\begin{theorem}\label{th:4.2}
Assume $\psi\in \RR$, i.e., $\psi_{\ri}=0$.
Let $t:=t_{1\rr}$, $\fs:=t_{2\ri}$,  and $s:=t_{3\rr}$ belonging to $\RR$, and let us consider 
\begin{equation}
\begin{pmatrix} d\varphi_1 \\ d \varphi_2 \\ d \varphi_3\end{pmatrix}
=(\cV_1,  \ii\cV_2, \cV_3)\begin{pmatrix} dt \\ d \fs \\ d s\end{pmatrix}.
\end{equation}
Then (\ref{4eq:loopMKdV2}) is reduced to the coupled FGMKdV equations,
\begin{eqnarray}
(\partial_{t}-\frac{1}{2}
(\lambda_{6}+3b_0+\frac{3}{4}(\partial_\fs\psi_\ri)^2)
          \partial_{s})\psi_\rr
           +\frac{1}{8}
\left(\partial_{s} \psi_\rr\right)^3
 +\frac{1}{4}\partial_{s}^3 \psi_\rr&=&0,
\label{eq:FGMKdV1r}
\\
(2\partial_{\fs}-\frac{1}{2}
(\lambda_{6}+3b_0-\frac{3}{4}(\partial_s\psi_\rr)^2)
          \partial_{\fs})\psi_\ri
           -\frac{1}{8}
\left(\partial_{\fs} \psi_\ri\right)^3
 +\frac{1}{4}\partial_{\fs}^3 \psi_\ri&=&0,
\label{eq:FGMKdV1i}
\end{eqnarray}
If $\partial_\fs \psi_\ri=0$ for a region, (\ref{eq:FGMKdV1r}) is further reduced to the focusing MKdV equation over $\RR$.
Note that $\partial_\fs \psi_{\ri}=\partial_{u_3,\rr} \psi_{\ri}$.
\end{theorem}

We recall that $\tgamma_i$ and thus $\cC$ represent the sheet from $\hX \to \PP^1$.

Theorem \ref{th:4.2} leads to the nice property that the conditions CI and CII in Condition \ref{cnd} are satisfied.
We explicitly describe the property following \cite{Ma24b}.
\begin{proposition}\label{4pr:reality_g3}
For configurations $(\ee^{\ii\varphi_i}, \ii8 \tgamma_i \tK_i\ee^{3\ii\varphi_i})_{i=1,2,3} \in S^3 \hX$ such that $\varphi_i \in  [\varphi_\fb^{[i-]}, \varphi_\fb^{[i+]}]\subset \RR$ for $(d s, d t) \in T^* \RR^2$, one-forms,
\begin{equation}
\begin{pmatrix} d \varphi_{1,\rr} \\ d \varphi_{2,\rr} \\
d \varphi_{3,\rr}\end{pmatrix}
= \cC\cV_3 d s=\begin{pmatrix}
\frac{2\tK_1\tgamma_1\cos(\varphi_2+\varphi_3)}{\sin(\varphi_2-\varphi_1)\sin(\varphi_3-\varphi_1)}\\
\frac{2\tK_2\tgamma_2\cos(\varphi_3+\varphi_1)}{\sin(\varphi_3-\varphi_2)\sin(\varphi_1-\varphi_2)}\\
\frac{2\tK_3\tgamma_3\cos(\varphi_1+\varphi_2)}{\sin(\varphi_1-\varphi_3)\sin(\varphi_2-\varphi_3)}
\end{pmatrix}d s ,
\label{eq:g3CIII}
\end{equation}
and $\cC\cV_1 d t$
form the two-dimensional real subspace in $T^* \hvarpi_x^{-1}
\prod_{a=1}^3[\varphi_\fb^{[a-]}, \varphi_\fb^{[a+]}]$. 
Thus the conditions CI and CII are satisfied, i.e.,
$\be_1 dt = \cL\cC (\cC \cV_1) dt$, and $\be_3 ds = \cL\cC (\cC \cV_3) ds$.
\end{proposition}

\begin{proof}
Basically the same as the proof in Proposition 5.3 in \cite{Ma24b}.
Since the region of $\varphi_a$ differs from \cite{Ma24b}, we give the proof.
Due to Lemma \ref{lm4.1}, 
$\be_1 ds = \cL\cC (\cC \cV_1) ds$ and $\be_2 dt = \cL\cC (\cC \cV_2) dt$.
Since the one forms in $\cC\cV_1 d s$ and  $\cC\cV_2 d t$ are holomorphic one-forms although $\cV_a$ is singular at the branch point, they are defined over $\hX$.
\end{proof}

Using the relation in Proposition \ref{4pr:reality_g3}, we have real two-dimensional subspaces in $S^3 \hX$ and $J_X$ respectively.
In other words, for $v_0 \in \CC^3$, we will deal with the real vector subspace $$
\displaystyle{\LL_{v_0}:=\left\{v_{t,s}:=\begin{pmatrix} t+s/2\\ -s/2\\ s\end{pmatrix}+v_0 := \int^s \be_1ds + \int^t \be_2dt+v_0 \Bigr|\ (t,s)\in \RR^2\right\}}\subset \CC^3.
$$
$\LL_{v_0}$ satisfies the condition CII,  i.e., $d u_{2, \ri}=d u_{3,\ri}$ $=d t_{2, \rr}= 0$.

\begin{proposition}\label{4pr:reality_g3a}
Assume Assumption \ref{Asmp}.

Let $(P_{a,0}=( \ee^{\ii\varphi_{a,0}},$ $ 8\ii \tgamma_a \tK(\varphi_{a,0})\ee^{3\ii \varphi_{a,0}})))_{a=1,2,3}$ be a point in $S^3 \hX$ where $\varphi_{a,0}\in [\varphi_\fb^{[a-]}, \varphi_\fb^{[a+]}]$, and set $\gamma_0 \in S^3\tX$ such that $\kappa_X \gamma_0 = (P_{1,0}, P_{2,0}, P_{3,0})$.
For $(t,s) \in \RR^2$,
\begin{equation}
  \begin{pmatrix} \varphi_1(t,s) \\ \varphi_2(t,s) \\ \varphi_3(t,s) 
  \end{pmatrix}
 :=  \left(\int^t_0 \cC\cV_1 d t + \int^s_0 \cC\cV_3 d s\right)+
\begin{pmatrix} \varphi_{1,0} \\ \varphi_{2,0} \\ \varphi_{3,0} 
  \end{pmatrix}
\label{eq:Xipnt}
\end{equation}
forms $\gamma(t,s)\in S^3 \tX$ and
\begin{equation}
\tSS_{\gamma_0}:=\{ \gamma(t,s) \ | \ (t,s) \in \RR^2\}\subset S^3\tX,
\label{eq:tS_varphi}
\end{equation}
if exists.
$\hkappa_X(\tSS_{\gamma_0})$ corresponds to a subspace in $\prod_{i=1}^3\Omega^1_i \subset S^3 \hX$ such that $\Phi(\gamma(t,s))$ to (\ref{eq:trans_Phi}) belongs to $[\varphi_\fb^{[a-]}, \varphi_\fb^{[a+]}]$.
Here $\Omega^1_i$ is a loop in $\hX$ whose image of $\hvarpi_x:(\ee^{\ii \varphi}, 8\ii K\ee^{3\ii \varphi})\mapsto \ee^{\ii \varphi})$ is the connected arc of the circle displayed in Figure \ref{fg:Fig01} of $\prod_{i=1}^3 [\varphi_\fb^{[i-]}, \varphi_\fb^{[i+]}]$. 

Further, let $v_0:=\tv(\gamma_0)$. 
Then $\tv(\gamma(t,s))=v_{t,s} \in \LL_{v_0}$ for $\gamma(t,s) \in \tSS_{\gamma_0}$.
\end{proposition}

\begin{proof}
Essentially the same as the proof in Proposition 5.4 in \cite{Ma24b}.
Since $[\varphi_\fb^{[a-]}, \varphi_\fb^{[a+]}]$ are disjoint, i.e.,
$[\varphi_\fb^{[a-]}, \varphi_\fb^{[a+]}]\cap[\varphi_\fb^{[b-]}, \varphi_\fb^{[b+]}]=\emptyset$ for $a\neq b$, the directions $\cV_1 d t$ and $\cV_3 d s$ are linearly independent and real valued.
Thus for a given $(\varphi_{a,0}) \in S^3 \RR$, $\hvarpi_x\hkappa_X\gamma(t,s)$ belongs to the unit circle or $\varphi_a(t,s)$ ($a=1,2,3$) are real.
From the construction, we have $\tv(\gamma(t,s))=v_{t,s}\in \LL_{v_0}$.
\end{proof}

Suppose that the subspace $\kappa_\hX(\tSS_{\gamma_0})\subset S^3 \Omega^1$ in Proposition \ref{4pr:reality_g3a} exists, and then we have the graph of $\tv^{-1}$ in $\CC^3\times S^3 \tX$,
\begin{equation}
\cG(\tv^{-1}|\LL_{\tv(\gamma_0)})=\left\{ \left(v_{t,s},\gamma(t,s) \right)\ | \ (t,s) \in \RR^2\right\}\subset \LL_{\tv(\gamma_0)} \times \tSS_{\gamma_0} 
\label{eq:LL_tS}
\end{equation}
by regarding $\LL_{\tv(\gamma_0)}$ as $\RR^2=\{(t,s)\}$.
At every point in the concerned subspace in $S^3\tX$, the meromorphic function $\psi$ satisfies the FGMKdV equations (\ref{eq:FGMKdV1r}) and (\ref{eq:FGMKdV1i}) as differential identities.
For a certain $\gamma =(\gamma_1, \gamma_2, \gamma_3)\in \tSS_{\gamma_0}$, we can express it by using (\ref{eq:trans_Phi}), i.e., $\varphi_a:=\Phi(\gamma_a)$ for $a=1,2,3$, and obtain a real valued one $\psi_{\rr}=2(\varphi_{1}+ \varphi_{2}+\varphi_{3})$.
$\cG(\tv^{-1}|\LL_{\tv(\gamma_0)}$ shows a solution of the FGMKdV equation including the time-development in $t$.

In this paper, however, we restrict $\tSS_{\gamma_0}$, $\LL_{v_0}$, and $\cG(\tv^{-1}|\LL_{v_0})$ to the following subspaces with $t=0$ for simplicity, in order to obtain loops beyond Euler's figure eight.
Let $v_0:=\tv(\gamma_0)$.
\begin{equation}
\begin{split}
&\tS_{\gamma_{0}}:=\{ \gamma(s):=\gamma(0,s) \ | \ s \in \RR\},\\
&L_{v_0}:=\left\{v_s:=\begin{pmatrix}s/2\\ -s/2\\ s\end{pmatrix}+v_0=\be_1 s +v_0 \Bigr|\ s\in \RR\right\}, \\
&\cG(\tv^{-1}|L_{v_0}):=\left\{ (v_s,\gamma(s))\ | \ s \in \RR\right\}\subset \tS_{\gamma_{0}}\times L_{v_0}\subset S^3 \tX \times \CC^3.
\label{eq:tS_varphi,t}
\end{split}
\end{equation}
Then using the graph space $(\gamma, \psi(\kappa_X(\gamma))) \in S^3 X \times \CC$, we have a graph structure from (\ref{eq:Xipnt}),
\begin{equation}
\Psi_{t=0}:=\left\{ \left(s, \psi_\rr(\kappa_X(\gamma(t=0,s))\right)\ | \ s \in \RR\right\}\subset \RR\times \RR.
\label{eq:Psit}
\end{equation}

More precisely speaking, we have the second theorem of this paper.

\begin{theorem}\label{pr:solgMKdV}
For $t=0$,  $\displaystyle{\begin{pmatrix}\varphi_{1}(0, s)\\ \varphi_{2}(0, s)\\ \varphi_{3}(0, s)\end{pmatrix}}$ $\varphi_a \in [\varphi_\fb^{[a-]}, \varphi_\fb^{[a+]}]$ ($a = 1, 2, 3$) for a solution $\gamma(0,s)$ of the differential equation (\ref{eq:g3CIII}) for $S^3 \hX$, we let $\psi_{\rr}(s)=2(\varphi_{1}(s)+ \varphi_{2}(s)$ $+\varphi_{3}(s))$.
Then $\psi_\rr(s, t)$ satisfies the FGMKdV equation (\ref{eq:FGMKdV1r}),
where the gauge field is $A(s,t)=(\lambda_{6}-3+\frac{3}{4}(\partial_{\fs}\psi_\ri(s,t))^2)/2$ given by 
\begin{equation}
\partial_{\fs}\psi_\ri(t=0, s)=(1,1,1)\cV_2.
\label{eq:g3CIIIi}
\end{equation}
\end{theorem}

\begin{proof}
The following is basically the same as the proof of Theorem 6.1 in \cite{Ma24b}, but that proof provides the algorithm for obtaining the global solutions for the hyperelliptic Riemann surface $\hX$, so we will explain it briefly.

From Theorems \ref{th:4.2} and \ref{pr:solgMKdV}, for $du_{3,\rr} = ds$, we have $\psi_\rr(s,t)$ and $\partial_{\fs}\psi_\ri$ in (\ref{eq:g3CIIIi}).

We show that a global solution of (\ref{eq:g3CIII}) exists as an orbit in $S^3 \Omega^1$. 
Let us prove this statement for the case of Figure \ref{fg:Fig01} (a) as follows.

We consider a primitive quadrature (\ref{eq:g3CIII}) for a real infinitesimal value $\delta s$.
We find $\delta \varphi_{a,\rr} =\cV_{1,a} \delta s$ for the real part of the entries of $\cK \cM$.
The orbit of $\varphi_a$ moves back and forth between the branch points $[\varphi_\fb^{[a-]}, \varphi_\fb^{[a+]}]=\hvarpi_x \Omega^1$ in (\ref{4eq:HEcurve_phi}) as in Figure \ref{fg:Fig01}.
Thus $\ee^{\ii \varphi_i}$ exists on the arc of the unit circle as in Figure \ref{fg:Fig01} (a) so that $(\ee^{\ii \varphi_a}, \tgamma K(\varphi_a))$ draws a loop $\Omega^1_a$ in $\hX$.

We note that at the branch point, the orbit of $\varphi_a$ turns the direction by changing the sign of $\tgamma_a$ so that $\tgamma_a \sin^2(\varphi_a)\delta \varphi_a/2K_a$ is positive as in Figure \ref{fg:Fig01};
it moves the different leaf of the Riemann surface with respect to the projection $\hvarpi_x:\hX \to \PP^1$ after passing the branch points.

Then we find an orbit in $S^3 \hX$ as a solution of (\ref{eq:g3CIII}) since $\cV_a$ is regular for disjoint $\varphi_a$, and $[\varphi_\fb^{[a-]}, \varphi_\fb^{[a+]}]$ and $[\varphi_\fb^{[b-]}, \varphi_\fb^{[b+]}]$ are disjoint if $a\neq b$.
Hence we prove it.
\end{proof}

\section{Numerical computations for closed curves}

Following the algorithm written in the proof of Theorem \ref{pr:solgMKdV}, we will numerically evaluate the hyperelliptic solutions of the FGMKdV equation in this section and demonstrate a generalization of Euler's elastica as mentioned in Introduction.

To estimate the magnitude of the gauge field $\partial_{s} \psi_{\ri}$, we introduce 
$\psi_\ri^\circ:=$
$\displaystyle{\int^s \partial_{\fs} \psi_{\ri} ds}$ noting $ \partial_{\fs} \psi_{\ri}= \partial_{u_3\rr} \psi_{\ri}$.
Further, we consider $\displaystyle{Z(s)=\int_{s_0}^s \ee^{\ii \psi_\rr(t=0, s)} ds}$ by fixing $t$ and $s_0$.

We numerically integrate the differential equation (\ref{eq:g3CIII}) following the algorithm explained in the proof of Theorem \ref{pr:solgMKdV}.
We study a generalization of Euler's elastica $C_Z:=\{(Z_\rr(s), Z_\ri(s)) \ |\ s \in [s_0, s_1]\}\subset \RR^2 =\CC$.

By tuning the parameters $k_a$ of the hyperelliptic curves and the initial conditions of $\varphi_{a,0}$, we can deal with closed generalized elastica given by the FGMKdV equation.
As we show in the following, we found two closed curves $C_Z$ as examples.

In other words, we have used the Euler's numerical quadrature method \cite{RLeV, Ma24c} for $\{(s, \hkappa_X\gamma(0,s))
\ |\ s \ge 0\}\subset \RR \times S^3 \tX$ and 
$\{(s, \psi_\rr((\gamma(0,s)))
\ |\ s \ge 0\}\subset \RR \times \RR$.
We will draw some graphs $C_Z$, $\{(s, \psi_\rr(s,t))\}_s$, $\{(s, \psi^\circ_\ri(s,t))\}_s$ and others based on the graph structure $\Psi_t$ in (\ref{eq:Psit}).

The first result is displayed in Figure \ref{fg:shape01}.
For the hyperelliptic curve given by $(k_1, k_2, k_3) =(2.12, 1.80, 1.0029)$, and the initial condition is 
$(\varphi_{1,0}, \varphi_{2,0}, \varphi_{3,0}) = (\varphi_{11}^\circ, \varphi_{21}^\circ, \varphi_{21}^\circ)$, where
$$
\varphi_{11}^\circ:=-0.8\varphi_{\fb3}^{++}, \quad
\varphi_{21}^\circ:=0.2\varphi_{\fb3}^{++}+0.8\varphi_{\fb2}^{++},\quad
\varphi_{31}^\circ:=0.2\varphi_{\fb3}^{+-}+0.8\varphi_{\fb2}^{+-}.
$$

\begin{figure}
\begin{center}

\includegraphics[width=0.45\hsize]{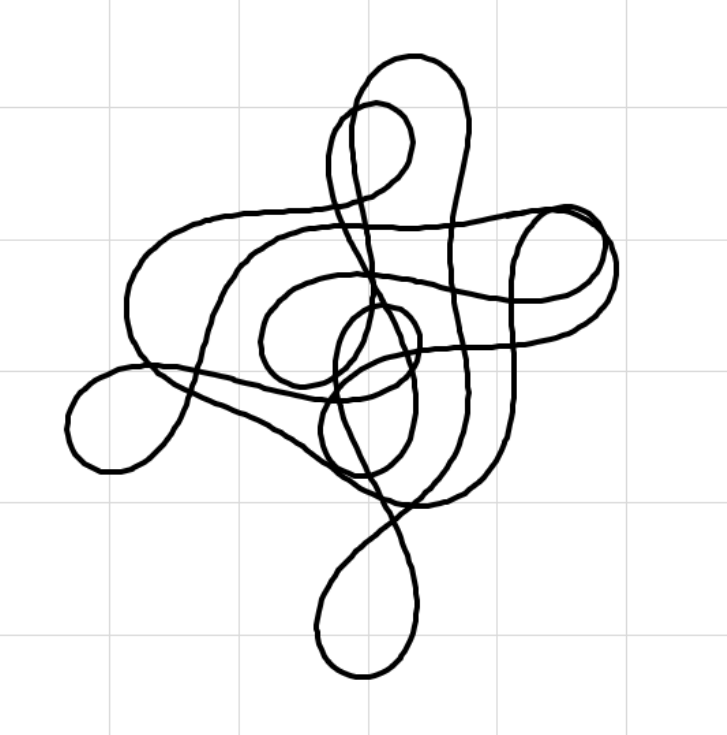}

(a)

\smallskip

\includegraphics[width=0.42\hsize,]{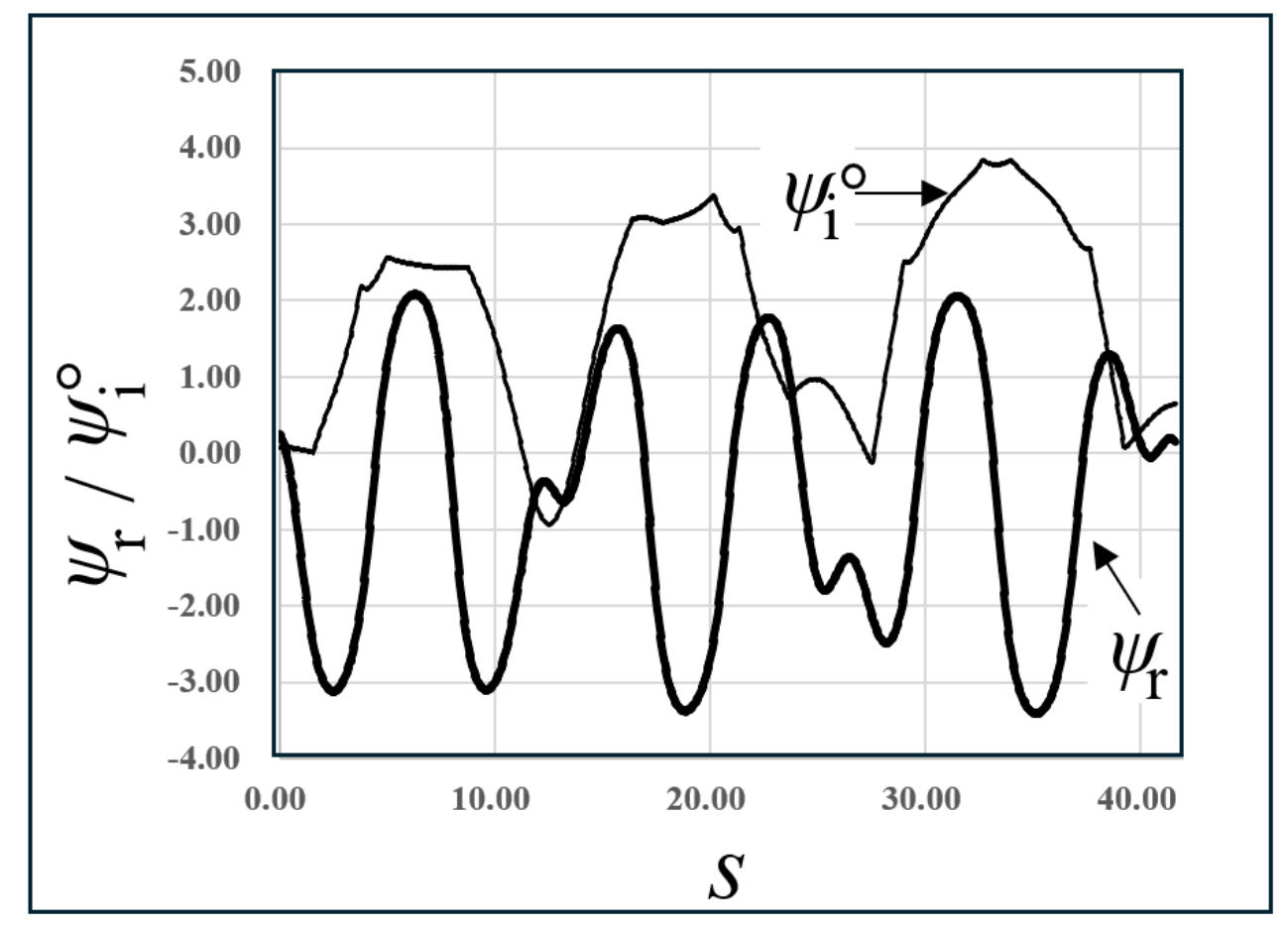}
\hskip 0.1\hsize
\includegraphics[width=0.42\hsize,]{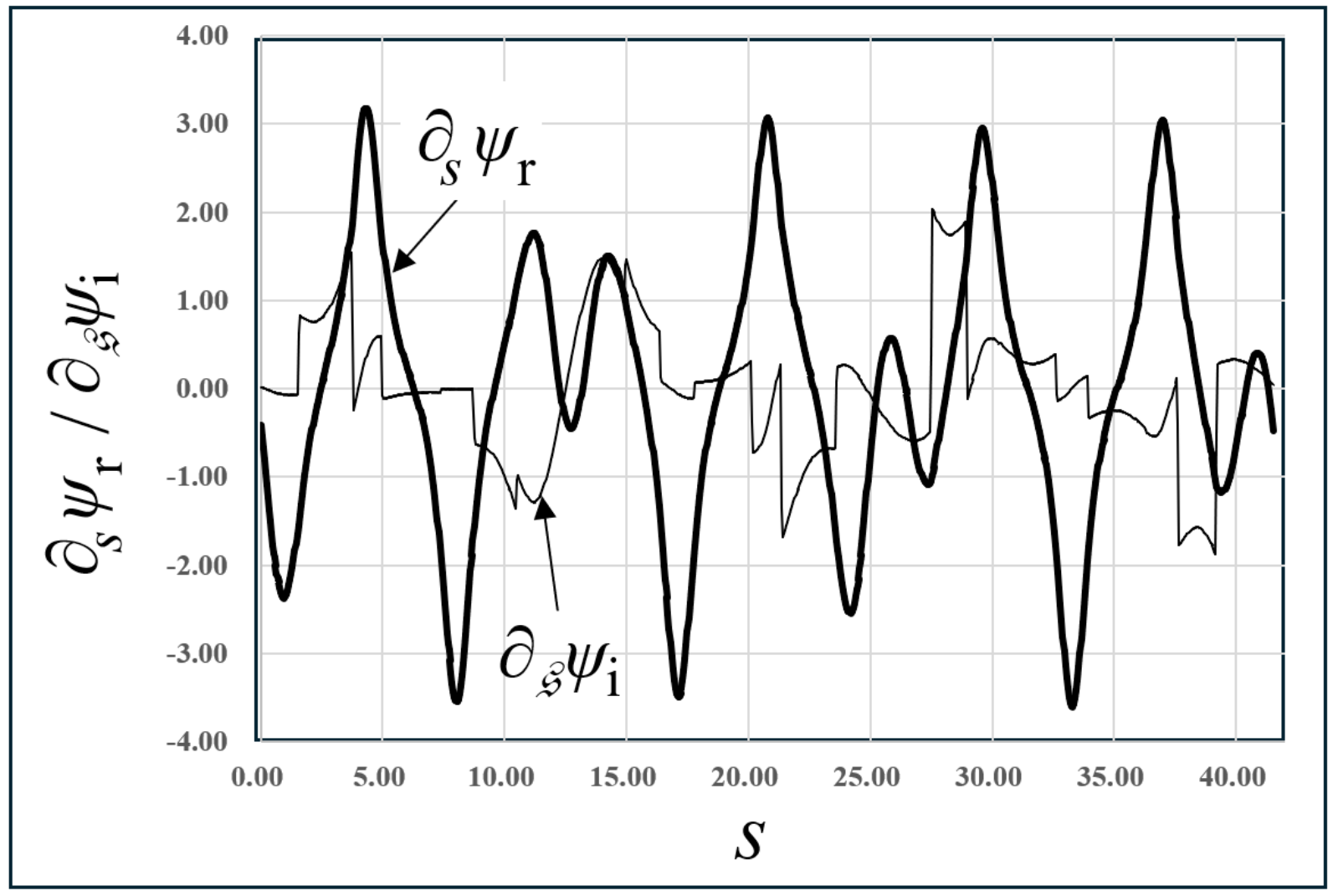}

(b) \hskip 0.4\hsize
(c)
\end{center}

\caption{
The solutions of (\ref{eq:Xipnt}) for $t=0$:
$(k_1, k_2, k_3) =(2.12, 1.80, 1.0029)$, and the initial condition is 
$(\varphi_1, \varphi_2, \varphi_3) = (\varphi_{11}^\circ, \varphi_{21}^\circ, \varphi_{21}^\circ)$.
(a): $C_Z$, (b): $\psi_\rr$ and $\psi_\ri^\circ$
and (c): $\partial_s\psi_\rr$ and $\partial_\fs\psi_\ri=\partial_{u_3\rr} \psi_\ri$.
}\label{fg:shape01}
\end{figure}

Figure \ref{fg:shape01} (a) shows $C_Z$, (b) display $\psi_\rr$ and $\psi_\ri^\circ$ while (c) shows $\partial_s\psi_\rr$ and $\partial_\fs\psi_\ri=\partial_{u_3\rr} \psi_\ri$.
$\psi_\rr(t=0,s)$ satisfies the FGMKdV equation (\ref{eq:FGMKdV1r}).
Since the value of $\partial_\fs\psi_\ri=\partial_{u_3\rr} \psi_\ri$ is not constant, $\psi_\rr(t=0,s)$ is a solution of the FGMKdV equation rather than the focusing MKdV equation.
Since SMKdV equation (\ref{4eq:SMKdV_k}) is a special case of the FGMKdV equation (\ref{eq:FGMKdV1r}), we can regard that Figure~\ref{fg:shape01} (a) is a generalization of Euler's elastica 

Figure \ref{fg:shape01} numerically shows that $Z(s+\ell) = Z(s)$, $\partial_s Z(s+\ell) = \partial_s Z(s)$, $\psi_\rr(s+\ell) = \psi_\rr(s)$ and $\partial_s\psi_\rr(s+\ell) = \partial_s\psi_\rr(s)$.
Since the tangential vector of the real plane curve $C_Z$ is given by $\partial_s Z=\ee^{\ii \psi}$, they mean that the curve $C_Z$ in Figure \ref{fg:shape01} is numerically continuous by the second order differential with respect to $s$.

\bigskip

\begin{figure}
\begin{center}

\includegraphics[width=0.7\hsize]{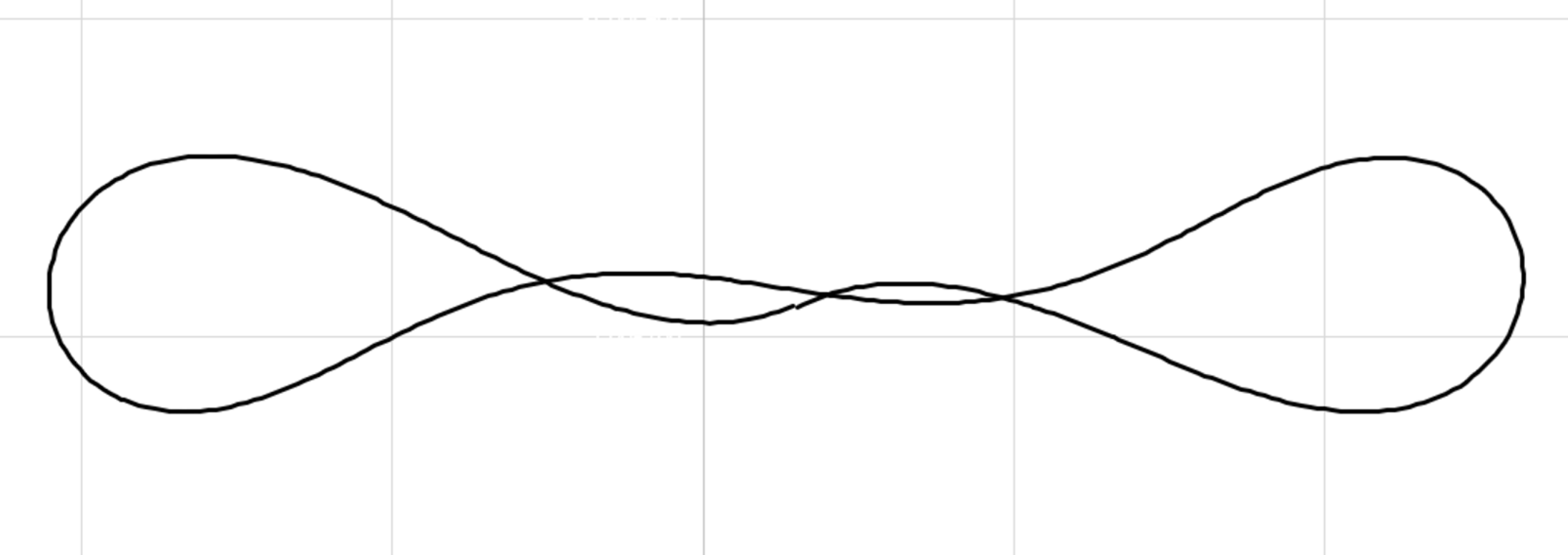}

(a)

\smallskip

\includegraphics[width=0.42\hsize,]{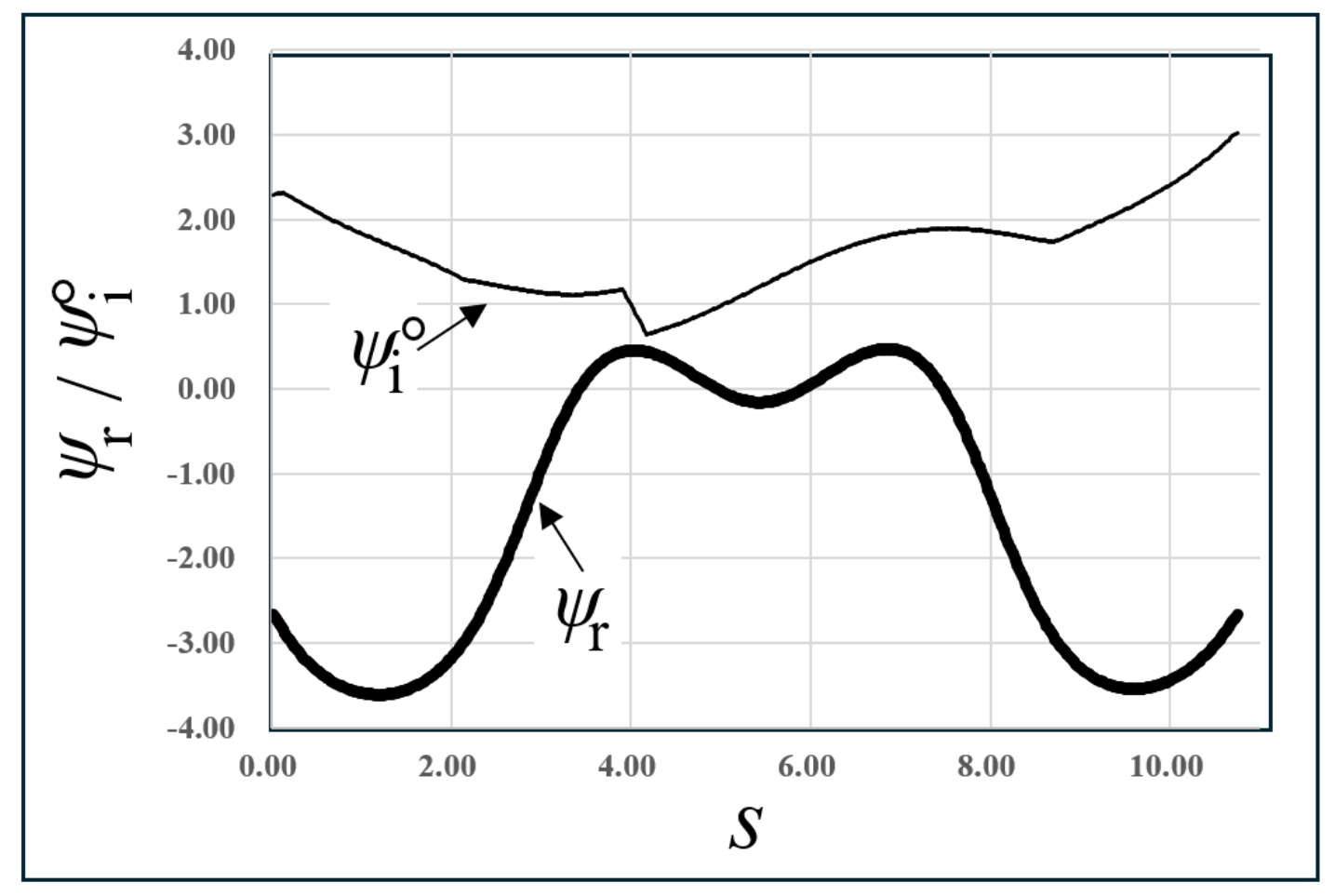}
\hskip 0.1\hsize
\includegraphics[width=0.42\hsize,]{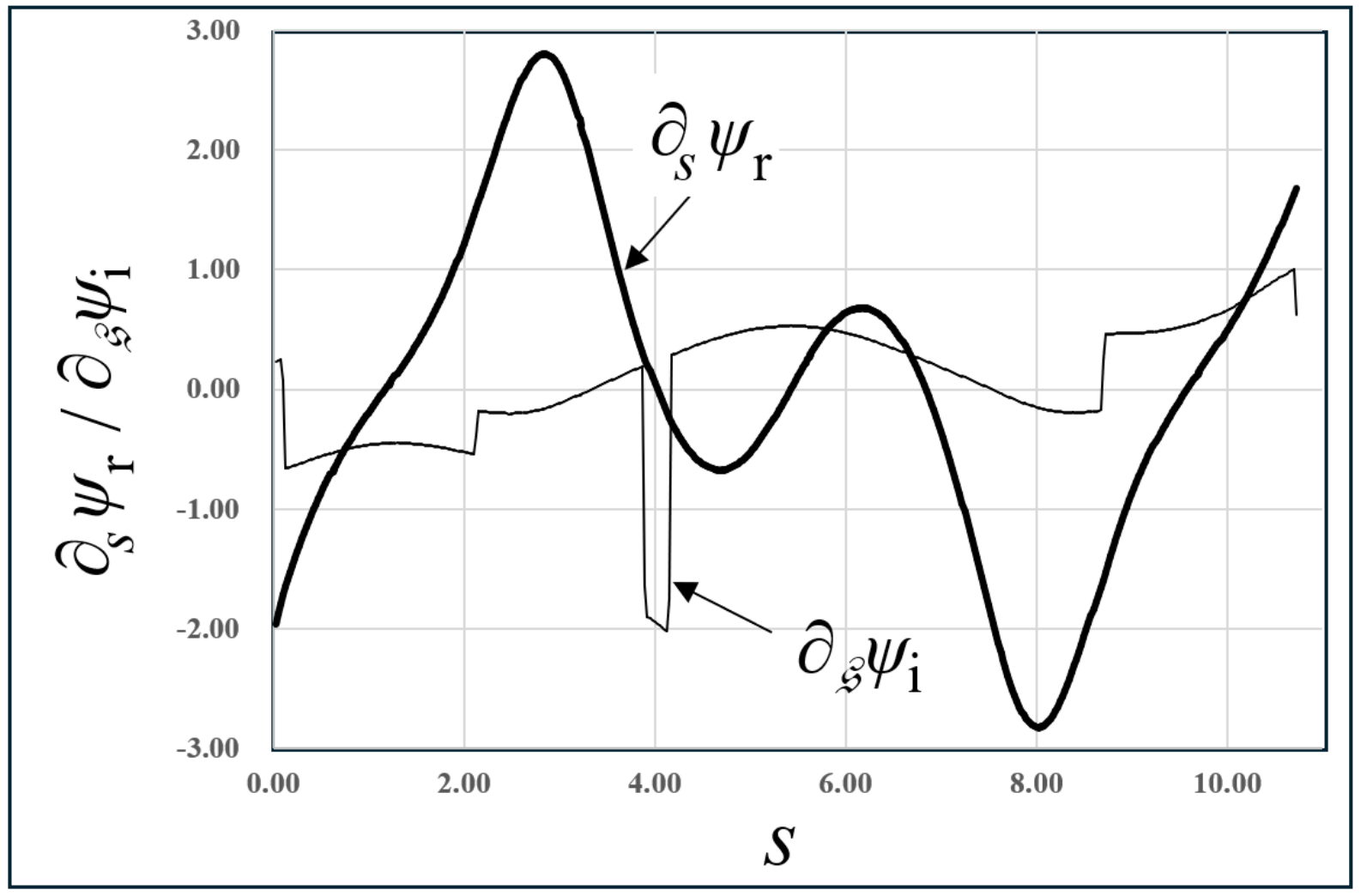}

(b) \hskip 0.4\hsize
(c)
\end{center}

\caption{
The solutions of (\ref{eq:Xipnt}) for $t=0$:
$(k_1, k_2, k_3) =(2.30, 1.85, 1.005)$, and the initial condition is 
$(\varphi_1, \varphi_2, \varphi_3) = (\varphi_{12}^\circ, \varphi_{22}^\circ, \varphi_{22}^\circ)$.
(a): $Z(s)$, (b): $\psi_\rr$ and $\psi_\ri^\circ$
and (c): $\partial_s\psi_\rr$ and $\partial_\fs\psi_\ri=\partial_{u_3\rr} \psi_\ri$.
}\label{fg:shape02}
\end{figure}

The second result is illustrated in Figure \ref{fg:shape02}.
We used the hyperelliptic curve given by $(k_1, k_2, k_3) =(2.30, 1.85, 1.005)$ and the initial condition is 
$(\varphi_{1,0}, \varphi_{2,0}, \varphi_{3,0}) = (\varphi_{12}^\circ, \varphi_{22}^\circ, \varphi_{22}^\circ)$, where
$$
\varphi_{11}^\circ:=-0.8\varphi_{\fb3}^{++}, \quad
\varphi_{21}^\circ:=0.2\varphi_{\fb3}^{++}+0.8\varphi_{\fb2}^{++},\quad
\varphi_{31}^\circ:=0.2\varphi_{\fb3}^{+-}+0.8\varphi_{\fb2}^{+-}.
$$
Figure \ref{fg:shape02} (a) shows these $\varphi_{a}$ and $\psi_\rr/2$.
As in Figure \ref{fg:shape02} (b) and (c) show $\psi_\rr$, $\psi^\circ_\ri$. $\partial_s \psi_\rr$ and $\partial_\fs\psi_\ri=\partial_{u_3\rr} \psi_\ri$.
Since the value of $\partial_\fs\psi_\ri=\partial_{u_3\rr} \psi_\ri$ is neither constant, $\psi_\rr(t=0,s)$ is also a hyperelliptic solutions of the FGMKdV equation rather than the focusing MKdV equation.
We can also regard that Figure~\ref{fg:shape02} (a) is a generalization of Euler's elastica

When we tuned the parameters, we consider only the condition that $Z(s+\ell) = Z(s)$, and $\partial_s Z(s+\ell) = \partial_s Z(s)$ since $\partial_s Z(s) =\ee^{\ii \psi_\rr(s)}$.
However, we did not impose $\partial_s\psi_\rr(s+\ell) = \partial_s\psi_\rr(s)$ in the second result as in Figure \ref{fg:shape02} (c).

Figure \ref{fg:shape02} numerically shows that $Z(s+\ell) = Z(s)$, $\partial_s Z(s+\ell) = \partial_s Z(s)$, and $\psi_\rr(s+\ell) = \psi_\rr(s)$.

Noting that the tangential vector of the real plane curve $C_Z$ is given by $\partial_s Z(s)=\ee^{\ii \psi(s)}$, we conclude that the curve $C_Z$ in Figure \ref{fg:shape01} is numerically continuous by the first order differential with respect to $s$. 
In fact, Figure \ref{fg:shape02} (a) and (b) numerically show that the tangent vectors at the ends of the curves agree and the positions coincide, i.e., $Z_\rr(s)=Z_\rr(s+\ell)$ and $Z_\ri(s)=Z_\ri(s+\ell)$.

Thus, we can regard $C_Z$ in Figure \ref{fg:shape02} (a) as a continuous closed real plane curve associated with the FGMKdV equation, just as we consider Euler's figure-eight as the closed real plane curve associated with the SMKdV equation (\ref{4eq:SMKdV_k}).
Figure \ref{fg:shape02} (a) is a generalization of the figure-eight in Figure \ref{fg:shape03}.

\section{Discussion and Conclusion}

In this paper, we showed a novel algebro-geometric method to obtain the hyperelliptic solutions of the FGMKdV equation (\ref{4eq:gaugedMKdV2}) as in Theorems \ref{th:4.2} and  \ref{pr:solgMKdV} associated with the focusing MKdV equation over $\CC$.
Then we have used the data of the hyperelliptic curves $X$ directly  instead of  the Jacobian $J_X$, and obtain the closed real plane curves as in Figures \ref{fg:shape01} and \ref{fg:shape02}.
We note that our algebraic study of the algebraic curves on two decades \cite{BEL97b, MP22, M24, Ma24a} allows the such treatment.

\bigskip

The ultimate goal of this project including this paper is to find the real solutions of the focusing MKdV equation of higher genus explicitly, i.e., hyperelliptic solutions of the FGMKdV equation with constant potential $A$, as mentioned in Introduction.
By revising the parameterization in \cite{Ma24b}, we found the parameterization  explicitly to express the real hyperelliptic solutions of the FGMKdV equation as in Theorems \ref{th:4.2} and  \ref{pr:solgMKdV}.
By numerically evaluating of the solutions, we obtain the closed real plane curves whose tangential angle satisfies the  FGMKdV equation as in Figures \ref{fg:shape01} and \ref{fg:shape02}.
They can be regarded as the generalization of Euler's elastica. 

In particular, Figure \ref{fg:shape02} is regarded as a generalization of a closed curve of Euler's figure-eight of elastica in Figure \ref{fg:shape03}.
Such a generalization has never been obtained, although there are so many studies on generalizations of Euler's elastica based only on the theory of elliptic functions whose shapes are similar to those of Euler's list of elasticae in 1744.
Figure \ref{fg:shape02} is a hyperelliptic solution of the FGMKdV equation with non-constant gauge potential $A$, which cannot be obtained in the framework of elliptic function theory.
There is no shape similar to Figure \ref{fg:shape02} in Euler's list.
In other words, we provided generalization of Euler's figure-eight in Figure \ref{fg:shape03}, as it took 280 years since 1744 \cite{Euler44}.
We should note that Figure \ref{fg:shape01} is also interpreted as a generalization of the figure-eight.

Furthermore, since our calculations seem to be many parameters associated with the closed generalized elasticae given by the FGMKdV equation. we will investigate the moduli space of the real closed plane curves associated with such hyperelliptic solutions of the FGMKdV equation as generalized elastica.

\bigskip

For our ultimate goal of finding the real analytical solutions of the focusing MKdV equation, i.e. the constant gauge or $\partial_{\fs} \psi_{\ri}$, this study shows the possibility of such solutions, since the set of solutions of the FGMKdV equation naturally contains the set of solutions of the focusing MKdV equation. 
It is expected that our results may exhibit some properties of the set related to the focusing MKdV equation, and simultaneously show that finding the focusing MKdV equation may require higher genus $g(>3)$ of hyperelliptic curves.

After we reach the ultimate goal, which is an exciting problem, as in \cite{Mat97, MP22, Ma24a}, the real analytic solutions of the focusing MKdV equation must represent the shapes of supercoiled DNA and provide the fascinating collaboration between life sciences and algebraic geometry.

\bigskip
\bigskip

%%%%%%%%%%%%%%%%%%%%%%%% referenc.tex %%%%%%%%%%%%%%%%%%%%%%%%%%%%%%
% sample references
% %
% Use this file as a template for your own input.
%
%%%%%%%%%%%%%%%%%%%%%%%% Springer-Verlag %%%%%%%%%%%%%%%%%%%%%%%%%%
%
% BibTeX users please use
% \bibliographystyle{}
% \bibliography{}
%

%\backmatter




\begin{thebibliography}{00}
%{ABCDE}



\bibitem{Baker97}
  \by{H. F. Baker}
  \book {Abelian functions}
% Abel's theorem
%and the allied theory of theta functions}
  \publ{Cambridge Univ. Press, Cambridge, 1995}
   Reprint of the 1897 original.

\bibitem{Baker98}
  \by{H.~F.~Baker}
  \paper{On the hyperelliptic sigma functions}
\jour{ Amer.~J.~Math.}
\vol{20} \yr{1898} \pages{301--384}.


\bibitem{Baker03}
  \by{H.~F.~Baker}
  \paper{On a system of differential equations
leading to periodic functions}
  \jour{Acta Math.}
  \vol{27}
  \yr{1903}
  \pages{135--156}



\bibitem{BEL97b}
\by{V.M. Buchstaber,  V.Z. Enolski\u{\i} and D.V. Le\u{\i}kin}
\paper{Kleinian functions, hyperelliptic Jacobians and applications}
\jour{Rev. Math. Math. Phys.}
\vol{10} (1997)  1--103.



\bibitem{CKPP}
\by{A. Chern, F. Kn\"oppel, F. Pedit and U. Pinkall}
\paper{Commuting Hamiltonian flows of curves in real space forms}
\jour{arXiv:1809.01394} (2018).

\bibitem{Euler44}
\by{L.~Euler} \book{Methodus Inveniendi Lineas Curvas Maximi
Minimive Proprietate Gaudentes} 1744.
%, Leonhardi Euleri Opera Omnia Ser. I vol. 14.



\bibitem{GoldsteinPetrich1}  
\by{R.~E.~Goldstein and D.~M.~Petrich}
\paper{The Korteweg-de Vries
hierarchy as dynamics of closed curves in the plane} 
\jour{Phys.~Rev.~Lett.} 
\vol{67} 
\yr{1991}
\pages{3203--3206}.

\bibitem{JGaray}
\by{J. Garay}
\paper{Extremals of the generalized Euler-Bernoulli energy and applications}
\jour{J. Geom. Symm. Phys.} \vol{12} \yr{2008} \pages{27--61}.


\bibitem{FarkasKra}
\by{H.~M.~Farkas and I.~Kra}
\book{Riemann Surfaces (GTM 71)}
\publ{Springer-Verlag} \publaddr{New York} 1991.

\bibitem{Koiso}
\by{N. Koiso}
\paper{Convergence towards an elastica in a Riemannian manifold}
\jour{Osaka J. Math.} \vol{37} (2000) \pages{467--487}.

\bibitem{KIW}
\by{K. Konno, Y. Ichikawa, M. Wadati}
\paper{New integrable nonlinear evolution equations}
\jour{J. Phys. Soc. Jpn}
\vol{47} \yr{1979} 1025-1026,

\bibitem{KIW2}
\by{K. Konno, Y. Ichikawa, M. Wadati}
{A loop soliton propagating along a stretched rope}
\jour{J. Phys. Soc. Jpn}
\vol{50} \yr{1981} 1025-1026.

\bibitem{KonnJeffrey}
\by{K. Konno, A. Jeffrey},
\paper{Some remarkable properties of two loop soliton solutions}
\jour{J. Phys. Soc. Jpn}
\vol{52} \yr{1983} 1-3.

\bibitem{Ishimori}
\by{Y. Ishimori}
\paper{On the modified Korteweg-de Vries soliton and the loop soliton}
\jour{J. Phys. Soc. Jpn}
\vol{50} \yr{1981}  2471-2472.



\bibitem{Langer}
\by{J. Langer}
\paper{Recursion in curve geometry}
\jour{New York J. Math.} \vol{5} \yr{1999} 25--51.

\bibitem{LangerPerline}
\by{J. Larger and R. Perline}
\paper{Curve motion inducing modified Korteweg-de Vries systems}
\jour{Phys. Lett. A} \vol{239} \yr{1998} \pages{36--40}.

\bibitem{Mat97}
\by{S.~Matsutani}
\paper{Statistical mechanics of elastica on a plane: 
origin of the MKdV hierarchy}
\jour{J. Phys. A: Math. \& Gen.}
\vol{31}  \yr{1998} \pages{2705-2725}.


\bibitem{Mat02c}
\by{S.~Matsutani} 
       \paper{Hyperelliptic solutions of modified Korteweg-de Vries equation
       of genus g:  essentials of Miura transformation}
       \jour{J.~Phys.~A:~Math.~\& Gen.} \vol{35} 
         \yr{2002} 4321--4333,

\bibitem{Mat02b}
\by{S. Matsutani} 
\paper{Hyperelliptic loop solitons with genus $g$: 
       investigation of a quantized elastica},
      \jour{J.~Geom.~Phys.} \vol{43} %\num{2-3} 
%     \yr{2002} 146-162.
     \yr{2002} 146.

\bibitem{Mat07}
\by{S.~Matsutani}
\textit{Reality conditions of loop solitons genus g}
\jour{Elec.~J.~Diff.~Eqns.} 
\vol{2007} 
\yr{2007}
\pages{1--12}.



\bibitem{Mat10}
\by{S. Matsutani} 
\paper{Euler's Elastica and Beyond}
\jour{J. Geom. Symm. Phys} \vol{17} \yr{2010}
\pages{45}.
%\jour{J. Geom. Symm. Phys} \vol{17} \yr{2010}
%\pages{45--86},

\bibitem{Ma24a}
\by{ S. Matsutani}
\paper{Statistical mechanics of elastica for the shape of supercoiled DNA: 
hyperelliptic elastica of genus three} 
\jour{Physica A} \vol{643} (2024) 129799 (11pages).

\bibitem{Ma24b}
\by{S.~Matsutani} 
\paper{On real hyperelliptic solutions of focusing modified KdV equation}	 \jour{Math. Phys. Ana. Geom.} \vol{23} (2024) \pages{19} (30pages).

\bibitem{Ma24c}
\by{S.~Matsutani} 
\paper{A numerical representation of hyperelliptic KdV solutions}
\jour{Comm. Nonlinear Sci. and Num. Sim.} \vol{138} (2024) 108259.


\bibitem{M24}
\by{ S. Matsutani}
\book{The Weierstrass sigma function in higher genus and applications to integrable equations} to appear as {\lq}Monographs in Mathematics{\rq} Springer 2025.

\bibitem{MP15}
\by{ S.~Matsutani and E.~Previato} 
\paper{The al function of a cyclic trigonal curve of genus three}
\jour{Collectanea Mathematica}
 \vol{66} 3, \yr{2015} 
\pages{311--349}.

\bibitem{MP16}
\by{ S. Matsutani and E. Previato}
 \paper{From Euler's elastica to the mKdV hierarchy, through the Faber 
polynomials}
\jour{J. Math. Phys.} \vol{57} (2016) 081519.


\bibitem{MP22}
\by{ S. Matsutani and E. Previato}
 \paper{An algebro-geometric model for the shape of supercoiled DNA}
\jour{Physica D} \vol{430} (2022) 133073.

\bibitem{Pinkall}
\by{U. Pinkall}
\paper{Hopf tori in $S^3$}
\jour{Invent. Math.} \vol{81} \yr{1985} 379--386.


\bibitem{P}
E. Previato,
\textit{
Geometry of the modified KdV equation},
in 
\textit{
LNP 424: Geometric and quantum aspects of 
integrable systems}
Ed. by G.~F.~Helminck, Springer 1993, 43-65.

\bibitem{RLeV}
\by{R. J. LeVeque}
\book{Finite difference methods for ordinary and partial differential equations: steady-state and time-dependent problems}
\publ{SIAM} Philadelphia 2007.



\bibitem{Wei54}
\by{K.~Weierstrass}
\paper{Zur Theorie der Abelschen Functionen}
\jour{J.~Reine Angew.~Math.} \vol{47} \yr{1854} 289--306.


\bibitem{ZakharovShabat}
\by{V. E. Zakharov and A. B. Shabat}
\paper{Exact theory of two-dimensional self-focusing and one-dimensional self-modulation of waves in nonlinear media}
\jour{Sov. Phys. JETP} 
\vol{84}, 62 (1972) \pages{62-69}


%\bibitem[Tr60]{Truesdell60}
%\by{C. Truesdell} \bibitem{T1}
% \book{The Rational Mechanics of Flexible or Elastic Bodies
%1638-1788}
% Leonhardi Euleri Opera Omnia Ser. II Vol. 11 part 2,
%Orell F\"ussli, Z\"urich, 1960.

%\bibitem[Tr68]{Truesdell68}
%\by{Truesdell C.}
%\book{ Essays in the History of Mechanics}
%\publ{Springer}
%\publaddr{New York}
% 1968.




\end{thebibliography}


%\begin{thebibliography}{00}


%\bibitem{}

%\end{thebibliography}
\end{document}